\documentclass[12pt]{iopart}
\usepackage{iopams}  
 \usepackage{graphicx}
 \usepackage{epstopdf}
 
\expandafter\let\csname equation*\endcsname\relax
\expandafter\let\csname endequation*\endcsname\relax
 \usepackage{amsmath}
 \usepackage{amssymb}
 \usepackage{mathrsfs}
 \usepackage{amsthm}
 \usepackage{bm}
 \usepackage{url}
 \usepackage[T1]{fontenc}
 \usepackage{csquotes}
 \MakeOuterQuote{"}
 \usepackage{cite}
 
 \newtheoremstyle{note}
 {\topsep/2} % ABOVE SPACE
 {\topsep/2} % BELOW SPACE
 {} % BODY FONT
 {\parindent} % INDENT (empty value is the same as 0pt)
 {\itshape} % HEAD FONT
 {.} % HEAD PUNCTUATION
 {5pt plus 1pt minus 1pt}% HEAD SPACE
 {}
 
 \theoremstyle{note}
 \newtheorem{theorem}{Theorem}
 \newtheorem{lemma}{Lemma}
 
 \newtheorem{corollary}{Corollary}
 \newtheorem{proposition}{Proposition}

 \theoremstyle{definition}

 \theoremstyle{remark}
 \newtheorem{remark}{Remark}

 \newcommand{\cnot}{\mathcal{U}_{\mrm{CNOT}}}

 \def\rE{\mathbb{E} }

 \newcommand{\mrm}[1]{\mathrm{#1}}

 \newcommand{\diag}{\operatorname{diag}}

 \newcommand{\supp}{\operatorname{supp}}

 \newcommand{\imply}{\mathrel{\Rightarrow}}

 \newcommand{\rmc}{\mathrm{c}}

 \newcommand{\rmr}{\mathrm{r}}

 \newcommand{\rmG}{\mathrm{G}}
 \newcommand{\rmL}{\mathrm{L}}

 \newcommand{\rmT}{\mathrm{T}}

 \newcommand{\caH}{\mathcal{H}}
 \newcommand{\caI}{\mathcal{I}}
 \newcommand{\caN}{\mathcal{N}}
 \newcommand{\caR}{\mathcal{R}}

 \newcommand{\scA}{\mathscr{A}}

 \newcommand{\lr}{{\mathcal{R}_\rmL}}

 \newcommand{\be}{\begin{equation}}
 \newcommand{\ee}{\end{equation}}
 \newcommand{\ba}{\begin{align}}
 \newcommand{\ea}{\end{align}}
 
 \def\<{\langle} %% overiding the original command \<
 \def\>{\rangle} %% overiding the original command \>
 \newcommand{\ket}[1]{| #1\>}

 \def\outer#1#2{|#1\>\<#2|} %% overiding the original command \outer

 \newcommand{\esref}[1]{\textup{(\ref{#1})}}
 \newcommand{\Esref}[1]{Equations~\textup{(\ref{#1})}}

 \newcommand{\thref}[1]{theorem~\ref{#1}}
 \newcommand{\Thref}[1]{Theorem~\ref{#1}}
 \newcommand{\thsref}[1]{theorems~\ref{#1}}

 \newcommand{\lref}[1]{lemma~\ref{#1}}
 
 \newcommand{\lsref}[1]{lemmas~\ref{#1}}

 \newcommand{\pref}[1]{proposition~\ref{#1}}
 \newcommand{\Pref}[1]{Proposition~\ref{#1}}

 \newcommand{\crref}[1]{corollary~\ref{#1}}
 \newcommand{\Crref}[1]{Corollary~\ref{#1}}

 \newcommand{\cref}[1]{conjecture~\ref{#1}}
 \newcommand{\Cref}[1]{Conjecture~\ref{#1}}

  \newcommand{\rcite}[1]{\cite{#1}}
  \newcommand{\rscite}[1]{\cite{#1}}
 
 \def\ps{\psi}
 \def\t{\theta}
 
 \newcommand{\proj}[1]{| #1\rangle\!\langle #1 |}
\begin{document}

\title{Coherence and entanglement measures based on R\'{e}nyi relative entropies}

\author{Huangjun Zhu$^1$, Masahito Hayashi$^{2,3}$, and Lin Chen$^{4,5}$}
\address{$^1$ Institute for Theoretical Physics, University of Cologne, Cologne 50937, Germany}

\address{$^2$ Graduate School of Mathematics, Nagoya University, 
	Nagoya, 464-8602, Japan}
\address{$^3$ Centre for Quantum Technologies, National University of Singapore, 
	3 Science Drive 2, 117542, Singapore}

\address{$^4$ School of Mathematics 
	and Systems Science, Beihang University, Beijing 100191, China}
\address{$^5$ International Research Institute for Multidisciplinary Science, 
	Beihang University, Beijing 100191, China}

\ead{hzhu1@uni-koeln.de, masahito@math.nagoya-u.ac.jp, and  linchen@buaa.edu.cn}

\bigskip
\begin{indented}
\item	\today
\end{indented}

\begin{abstract}
We study systematically  resource measures of
coherence  and entanglement based on R\'enyi relative entropies, which include the  logarithmic robustness of coherence, geometric coherence, and  conventional relative entropy of coherence together with their entanglement analogues. First, we show that each R\'enyi relative entropy of coherence is equal to the corresponding R\'enyi relative entropy of entanglement for any maximally correlated state. By virtue of this observation, we 
establish a simple operational connection between entanglement
measures and coherence measures based on R\'enyi relative entropies. We then prove that all these coherence measures, including the logarithmic robustness of coherence, are additive. Accordingly, all these entanglement measures are additive for maximally correlated states.
In addition, we  derive analytical formulas for R\'enyi relative entropies of entanglement of maximally correlated states and bipartite pure states, which reproduce a number of classic results on the relative entropy of entanglement and logarithmic robustness of entanglement in a unified framework.
Several nontrivial bounds for R\'enyi relative entropies of coherence (entanglement) are further derived, which improve over results known previously. 
Moreover, we determine all states whose relative entropy of coherence is equal to the logarithmic robustness of coherence. As an application, we provide an upper bound for the exact coherence distillation rate, which is saturated  for pure states.

\end{abstract}

% Uncomment for PACS numbers
%\pacs{00.00, 20.00, 42.10}
%
% Uncomment for keywords
%\vspace{2pc}

\section{Introduction}

Quantum coherence is a root of many nonclassical phenomena and a valuable resource for quantum information processing. Recently, the resource theory of coherence was established in \rscite{Aber06,BaumCP14, WintY16} and stimulated increasing attention in the quantum information community; see \rscite{StreAP17,HuHPZ17} for a review. It turns out that this resource theory is closely related to the well-established resource theory of entanglement \cite{StreSDB15,YuanZCM15,DuBG15,DuBQ15,KillSP16,ZhuMCF17,ZhuHC17,RanaPWL16,DanaDMW17}, which plays a crucial role in the development of coherence theory. Understanding the connections between the two resource theories is a focus of ongoing research.

Recently, Streltsov et al. showed that coherence with respect to a reference basis can be converted to entanglement by incoherent operations acting on the system and an incoherent ancilla \cite{StreSDB15}. Moreover, the maximum entanglement generated in this way defines a coherence measure. Surprisingly, this mapping can establish a one-to-one correspondence between many useful entanglement measures and coherence measures, including those based on the relative entropy, fidelity, and convex-roof construction~\cite{StreSDB15,ZhuMCF17}. Although not so obvious, the $l_1$-norm of coherence \cite{BaumCP14,RanaPWL16} turns out to be the analogue of the negativity under this mapping~\cite{ZhuHC17}.

Despite these progresses, it is still not clear what coherence  measures in general can be derived from entanglement  measures in a natural way. A case in point is the family of coherence  measures based on R\'enyi relative entropies \cite{MullDSF14,WildWY14,Haya17book,ChitG16a,ShaoLLX16}, which  includes three of the most important coherence measures, namely, relative entropy of coherence \cite{Aber06,BaumCP14} (equal to the distillable coherence \cite{WintY16,YuanZCM15}), logarithmic robustness of coherence \cite{NapoBCP16,PianCBN16,ChitG16a,RanaPWL16}, and geometric coherence \cite{StreSDB15}. Their  entanglement analogues are equally important in the resource theory of entanglement  \cite{HoroHHH09}. Although these resource measures  have been studied extensively, most previous works 
focus on individual measures separately, without studying the connections between them, which leads to severe limitation on our understanding about this subject.

In this paper we explore the connections between entanglement and coherence by studying systematically resource measures based on R\'enyi relative entropies. First, we show that 
 R\'enyi relative entropies of coherence and entanglement are equal to the corresponding R\'enyi conditional entropies for maximally correlated states. Interestingly, the same conclusion holds for three variants of entanglement measures based on separable states, positive-partial-transpose (PPT) states, and nondistillable states, respectively. By virtue of  this observation, 
 we show that each R\'enyi relative entropy of coherence  is equal to the  maximum of the corresponding R\'enyi relative entropy of entanglement generated by incoherent operations acting on the system and an incoherent ancilla. The generalized CNOT gate turns out to be the common optimal incoherent operation.
In this way, we set an operational one-to-one mapping between entanglement
measures and coherence measures based on R\'enyi relative entropies, which complements a similar mapping between measures based on the convex roof~\cite{ZhuMCF17}.

We then prove that all R\'enyi relative entropies of coherence, including the logarithmic robustness of coherence, are additive. As an implication, all R\'enyi relative entropies of entanglement are additive for maximally correlated states. 
In addition, 
we derive several nontrivial bounds on R\'enyi relative entropies of coherence and the robustness of coherence, which significantly improve over bounds known before. In particular, our study  shows that the logarithmic $l_1$-norm of coherence is a universal upper bound for all  R\'enyi relative entropies of coherence. Similar results apply to  R\'enyi relative entropies of entanglement of maximally correlated states.
Moreover, we  derive analytical formulas for R\'enyi relative entropies of entanglement of maximally correlated states and bipartite pure states, which reproduce a number of classic results on the relative entropy of entanglement and logarithmic robustness of entanglement in a unified framework.

Furthermore, we clarify the relations between different R\'enyi relative entropies of coherence and determine all states whose relative entropy of coherence (or distillable coherence) is equal to the logarithmic robustness of coherence. It turns out that for these states all R\'enyi relative entropies of coherence coincide with the  relative entropy of  coherence.
To achieve this goal, we determine the condition under which R\'enyi relative entropies  are independent of the order parameter, note that they are usually monotonically increasing with this parameter.

As an application, we provide an upper bound for the exact coherence distillation rate, which is saturated for pure states.
It turns out that for pure states this rate remains the same under three distinct classes of operations, namely, strictly incoherent operations, incoherent operations, and incoherence-preserving operations. This result parallels a similar result on exact entanglement distillation \cite{HayaKMM03,Haya06,Haya17book}, which further strengthens the connection between the resource theory of coherence and that of entanglement. 
In addition, we derive a necessary condition under which the exact coherence distillation rate is equal to the distillable coherence, thereby clarifying the relation between exact coherence distillation and approximate distillation with vanishing error asymptotically. 
Besides, the results presented here play a crucial role in studying secure random number generation via incoherent operations \cite{HayaZC17}.

The rest of this paper is organized as follows. In \sref{sec:Pre} we review the basic concepts and known results about R\'enyi relative entropies together with entanglement measures and coherence measures based on them. In \sref{sec:Connect} we establish an operational one-to-one mapping between entanglement measures and coherence measures based on R\'enyi relative entropies and thereby derive R\'enyi relative entropies of entanglement of maximally correlated states. In \sref{sec:Additivity} we prove the additivity of R\'enyi relative entropies of coherence and the logarithmic robustness of coherence. In \sref{sec:bounds} we derive several nontrivial upper and lower bounds for  R\'enyi relative entropies of coherence.
In \sref{sec:relationR} we investigate the relations between different R\'enyi relative entropies. 
In \sref{sec:relationRREC} we clarify the relations between different R\'enyi relative entropies of coherence. In \sref{sec:ExactDistill} we provide an upper bound for the exact coherence distillation rate, which is saturated for pure states. \Sref{sec:Sum} summarizes this paper.

\section{\label{sec:Pre}Preliminaries}
In this section we review the basic concepts and known results about two types of R\'enyi relative entropies together with entanglement measures and coherence measures based on them.
A few new results are added for completeness. 

\subsection{\label{sec:RRECE}R\'enyi relative entropies and conditional entropies}
The relative entropy between two density matrices $\rho$ and $\sigma$ on a given Hilbert space $\caH$
reads
\begin{equation}
S(\rho\|\sigma):=\tr(\rho\ln\rho)-\tr(\rho\ln\sigma)=-S(\rho)-\tr(\rho\ln\sigma),
\end{equation}
where "$\ln$" denotes the natural logarithm and $S(\rho):=-\tr(\rho\ln\rho)$ denotes the von Neumann entropy of $\rho$. Although we choose the natural logarithm in this paper, 
except for \sref{sec:relationR}, however,  the choice of the base for the logarithm does not affect our results explicitly as long as "$\exp$" and "$\log$" take on the same base.
The relative entropy $S(\rho\|\sigma)$ reduces to the relative entropy between two probability distributions when both $\rho$ and $\sigma$ are diagonal with respect to a reference basis. 

As  generalization, consider two types of 
R\'enyi relative entropies  \cite{MullDSF14,WildWY14}
\cite[Section 3.1]{Haya17book}
\begin{align}
S_\alpha(\rho\|\sigma):= \frac{1}{\alpha-1}\ln \tr\bigl( \rho^{\alpha}\sigma^{1-\alpha}\bigr), \quad
\underline{S}_\alpha(\rho\|\sigma):= 
\frac{1}{\alpha-1}\ln \tr\bigl( \sigma^{\frac{1-\alpha}{2\alpha}}\rho\sigma^{\frac{1-\alpha}{2\alpha}} \bigr)^{\alpha}, \label{eq:RRE}
\end{align}
where  $\alpha\geq0$ is known as the order parameter. The power of a positive operator is understood as the power on its support.
The second argument $\sigma$ in $S_\alpha(\rho\|\sigma)$ and $\underline{S}_\alpha(\rho\|\sigma)$ can be generalized to positive operators. These R\'enyi relative entropies have wide applications in quantum information processing~\cite{Haya17book} and
have operational interpretations in connection with quantum hypothesis testing~\cite{MosoO15}.

In the cases $\alpha=0, 1, \infty$, 
the definitions of $S_\alpha(\rho\|\sigma)$ and $\underline{S}_\alpha(\rho\|\sigma)$ above are understood as proper limits, all of which are well defined. Hence, the order parameter~$\alpha$ for both types of  R\'{e}nyi relative entropies can be regarded to run from $0$ to $\infty$. To be concrete,
\begin{equation}
S_0(\rho\|\sigma)=\lim_{\alpha\to 0}S_\alpha(\rho\|\sigma)=-\ln\tr(\Pi_\rho\sigma),
\end{equation} 
where $\Pi_\rho$ is the projector onto the support of $\rho$; the limit $\underline{S}_0(\rho\|\sigma)=\lim_{\alpha\to 0}\underline{S}_\alpha(\rho\|\sigma)$ is derived in \cite{AudeD15}, but is not needed here. Both $S_\alpha(\rho\|\sigma)$ and $\underline{S}_\alpha(\rho\|\sigma)$ approach $S(\rho\|\sigma)$ in the limit $\alpha\rightarrow 1$.  The limits $\lim_{\alpha\to \infty}{S}_\alpha(\rho\|\sigma)$ and
$\lim_{\alpha\to \infty}\underline{S}_\alpha(\rho\|\sigma)$  are written as  ${S}_\infty(\rho\|\sigma)$ and $\underline{S}_\infty(\rho\|\sigma)$, respectively. The latter $\underline{S}_\infty(\rho\|\sigma)$ is known as the  max relative entropy \cite{Datt09, Datt09b,DupuKFR13,MullDSF14} and can be expressed as 
\begin{align}
\underline{S}_\infty(\rho\|\sigma)&=\min \{\ln\lambda | \lambda \sigma \ge \rho \}.
\end{align}
The following two special cases of $\underline{S}_\alpha(\rho\|\sigma)$ are also useful to the current study,
\begin{align}
\underline{S}_{1/2}(\rho\|\sigma)&=-\ln F(\rho,\sigma), \label{eq:MinRE}\\
\underline{S}_2(\rho\|\sigma)&=\ln\tr\bigl[\bigl(\sigma^{-1/4}\rho\sigma^{-1/4}\bigr)^2\bigr]=\ln\tr\bigl(\sigma^{-1/2}\rho\sigma^{-1/2}\rho\bigr),
\end{align}
where $F(\rho,\sigma):=\bigl(\tr\sqrt{\sigma^{1/2}\rho\sigma^{1/2}}\bigr)^2$ denotes the fidelity between $\rho$ and $\sigma$.  The two relative entropies  $\underline{S}_{1/2}(\rho\|\sigma)$ and $\underline{S}_2(\rho\|\sigma)$ are known as the
min relative entropy and collision relative entropy, respectively \cite{Datt09, Datt09b,DupuKFR13,MullDSF14}.

According to the Araki-Lieb-Thirring inequality \cite{LiebT76,Arak90} and the result in
\rcite{Hiai94}, the two types of R\'enyi relative entropies defined in \eref{eq:RRE}  satisfy the following inequality
\cite{MullDSF14,WildWY14}\cite[Section 3.1]{Haya17book}
\begin{align}
S_\alpha(\rho\|\sigma) \ge
\underline{S}_\alpha(\rho\|\sigma)\quad \forall \alpha\in [0,\infty].\label{eq:RREabOrder}
\end{align}
When  $\alpha\in (0,\infty)$ with $\alpha\neq 1$,  the inequality is  saturated if and only if (iff) $\rho$ and $\sigma$ commute  \cite{MosoO15}. Both $S_\alpha(\rho\|\sigma)$ and $\underline{S}_\alpha(\rho\|\sigma)$
are monotonically increasing (means nondecreasing in this paper) with $\alpha$.
Similar to $S(\rho\|\sigma)$, the R\'enyi relative entropy
${S}_\alpha(\rho\|\sigma)$ satisfies the data-processing inequality
for $\alpha 
\in [0,2]$~\cite{Petz86},
and 
$\underline{S}_\alpha(\rho\|\sigma)$ satisfies the data-processing inequality for 
$\alpha \in [\frac{1}{2},\infty]$ \cite{MullDSF14,Beig13,FranL13,WildWY14}\cite[lemma 3.1]{Haya17book}. In other words, these R\'enyi relative entropies are contractive under  any completely positive and trace-preserving (CPTP) map $\Lambda$. 
More precisely, we have 
\begin{align}
S_\alpha(\Lambda(\rho)\|\Lambda(\sigma))&\leq S_\alpha(\rho\|\sigma)\quad \forall \alpha 
\in [0,2], \label{eq:DPI-RREa}\\
\underline{S}_\alpha(\Lambda(\rho)\|\Lambda(\sigma))&\leq \underline{S}_\alpha(\rho\|\sigma) \quad \forall \alpha \in \Bigl[\frac{1}{2},\infty\Bigr]. \label{eq:DPI-RREb}
\end{align}

In addition, $\exp[(\alpha-1)S_\alpha(\rho\|\sigma)]$ is jointly convex for $\alpha 
\in (1,2]$ and jointly concave for $\alpha 
\in [0,1)$; by contrast,  $\exp[(\alpha-1)\underline{S}_\alpha(\rho\|\sigma)]$ is jointly convex for $\alpha \in (1,\infty]$ and jointly concave for $\alpha \in [\frac{1}{2},1)$  \cite{MullDSF14,Haya17book2}. To see this, let $\rho_1,\rho_2,\sigma_1,\sigma_2$ be four arbitrary quantum states on $\caH$ and $0\leq \lambda\leq 1$. Consider the two 
states
\begin{align}
\rho:=\begin{pmatrix}
\lambda \rho_1 & 0 \\
0 & (1-\lambda) \rho_2 
\end{pmatrix}, \quad
\sigma:=\begin{pmatrix}
\lambda \sigma_1 & 0 \\
0 & (1-\lambda) \sigma_2 
\end{pmatrix}
\end{align}
on the composite system $\mathbb{C}^2\otimes {\cal H}$. Taking the partial trace over the first subsystem yields
\begin{align}
&\lambda \exp\bigl[(\alpha-1){S}_\alpha(\rho_1\|\sigma_1)\bigr]
+
(1-\lambda) \exp\bigl[(\alpha-1){S}_\alpha(\rho_2\|\sigma_2)\bigr] 
= \exp\bigl[(\alpha-1){S}_\alpha(\rho\|\sigma)\bigr]\nonumber\\
&\geq
\exp\bigl[(\alpha-1){S}_\alpha(
	\lambda \rho_1 + (1-\lambda) \rho_2
	\|
	\lambda \sigma_1 + (1-\lambda) \sigma_2)\bigr] \quad \forall \alpha\in (1,2],\label{eq:H7}
\end{align}
where the inequality follows from the data-processing inequality \eref{eq:DPI-RREa} and the fact that the partial trace is a CPTP map.
Therefore, $\exp[(\alpha-1)S_\alpha(\rho\|\sigma)]$ is jointly convex for $\alpha 
\in (1,2]$. The joint convexity of $\exp[(\alpha-1)\underline{S}_\alpha(\rho\|\sigma)]$ for $\alpha \in (1,\infty]$ follows from the same reasoning. The joint concavity
of $\exp[(\alpha-1)S_\alpha(\rho\|\sigma)]$ for $\alpha \in [0,1)$ and $\exp[(\alpha-1)\underline{S}_\alpha(\rho\|\sigma)]$ for $\alpha \in [\frac{1}{2},1)$ can also be proved in a similar manner.

Next, we turn to conditional entropies constructed from R\'enyi relative entropies. Given a bipartite state $\rho$ shared by Alice (A) and Bob (B), the conditional entropy of A given B have three equivalent definitions,
\begin{align}
H(A|B)_{\rho}:=&S(\rho_{AB})-S(\rho_B) =- S(\rho_{AB} \| I_A \otimes \rho_B)
=-\min_{\sigma_B} S( \rho_{AB} \| I_A \otimes \sigma_B),\label{eq:ConEntropy}
\end{align} 
where $\rho_{AB}=\rho$ (the subscripts are omitted if there is no confusion), $\rho_B=\tr_A(\rho)$, 
$I_A$ denotes the identity on $\caH_A$, and the minimization is taken over all quantum states $\sigma_B$ on $\caH_B$. 
However, only the second and third definitions  above admit meaningful generalizations, which produce four types of R\'enyi conditional entropies \cite{MullDSF14,TomaBH14},
\begin{align}
H_\alpha^{\downarrow}(A|B)_{\rho}
&:=- {S}_\alpha(\rho_{AB} \| I_A \otimes \rho_{B}), \quad &
H_\alpha^{\uparrow}(A|B)_{\rho}
&:=-\min_{\sigma_B} {S}_\alpha(\rho_{AB} \| I_A \otimes \sigma_{B}), \\
\overline{H}_\alpha^{\downarrow}(A|B)_{\rho}
&:=- \underline{S}_\alpha(\rho_{AB} \| I_A \otimes \rho_B),  \quad
&\overline{H}_\alpha^{\uparrow}(A|B)_{\rho}
&:=-\min_{\sigma_B} \underline{S}_\alpha(\rho_{AB} \| I_A \otimes \sigma_B). 
\end{align}
By definitions and  the inequality $S_\alpha(\rho\|\sigma) \ge
\underline{S}_\alpha(\rho\|\sigma)$ in \eref{eq:RREabOrder}, these conditional entropies satisfy
\begin{equation}\label{eq:RREineqSimple}
\begin{aligned}
H_\alpha^{\downarrow}(A|B)_{\rho}
&\leq
H_\alpha^{\uparrow}(A|B)_{\rho},\quad & \overline{H}_\alpha^{\downarrow}(A|B)_{\rho}
&\leq \overline{H}_\alpha^{\uparrow}(A|B)_{\rho},\\
H_\alpha^{\downarrow}(A|B)_{\rho}
& \leq \overline{H}_\alpha^{\downarrow}(A|B)_{\rho},\quad & 
H_\alpha^{\uparrow}(A|B)_{\rho}
&\leq \overline{H}_\alpha^{\uparrow}(A|B)_{\rho}.
\end{aligned}
\end{equation}

The conditional entropy $H_\alpha^{\uparrow}(A|B)_{\rho}$ has a closed formula according to \rcite{TomaBH14},
\begin{equation}\label{eq:CEaformula}
H_\alpha^{\uparrow}(A|B)_{\rho}=\frac{\alpha}{1-\alpha}\ln\tr \left\{\left[\tr_A (\rho_{AB}^\alpha)\right]^{1/\alpha} \right\}.
\end{equation}
When  $\rho$ is a classical-quantum state, i.e., it has the form
$\rho=\sum_{a} P_A(a) (|a\rangle \langle a| \otimes \rho_{B|a})$,
the quantity $\exp\bigl[-\overline{H}_{\infty}^{\uparrow}(A|B)_{\rho}\bigr]$ 
expresses the optimal probability of guessing correctly the classical information concerning A from the quantum system B \cite[theorem 1]{KoniRS09}.

When $\rho=\rho_A\otimes \rho_B$ is a tensor product, straightforward calculation shows that the four types of R\'enyi conditional entropies coincide with each other,
\begin{align}\label{eq:RCEproductS}
\overline{H}_\alpha^{\uparrow}(A|B)_{\rho} &=\overline{H}_\alpha^{\downarrow}(A|B)_{\rho}= H_\alpha^{\uparrow}(A|B)_{\rho}=H_\alpha^{\downarrow}(A|B)_{\rho}=-S_\alpha(\rho_A\|I_A)=S_\alpha(\rho_A),
\end{align}
where 
\begin{equation}
S_\alpha(\rho_A):=\frac{1}{1-\alpha}\tr(\rho_A^\alpha)
\end{equation}
is the R\'enyi $\alpha$-entropy of $\rho_A$.

When $\rho$ is a tripartite pure state shared by A, B and E (Eve), R\'enyi conditional entropies obey the following duality relations.
\begin{proposition}[\protect{\cite{MullDSF14}\cite{Beig13}\cite{TomaBH14}\cite[theorem 5.13]{Haya17book}}]\label{pro:duality}
	\begin{align}
	H_\alpha^{\downarrow}(A|E)_{\rho}
	+{H}_\beta^{\downarrow}(A|B)_{\rho}&=0,
	\label{eq:dualitya}\\
	\overline{H}_\alpha^{\uparrow}(A|E)_{\rho}+
	\overline{H}_\beta^{\uparrow}(A|B)_{\rho} &=0, \label{eq:dualityb} \\
	\overline{H}_\alpha^{\downarrow}(A|E)_{\rho}
	+{H}_\beta^{\uparrow}(A|B)_{\rho}&=0,
	\label{eq:dualityab}
	\end{align}
	where \eref{eq:dualitya} holds for $\alpha,\beta \in[0,2]$ with $\alpha+\beta=2$, 
	\eref{eq:dualityb} holds for $\alpha,\beta \in[\frac{1}{2},\infty]$ with $ \frac{1}{\alpha}+\frac{1}{\beta}=2$,  
	and  \eref{eq:dualityab} holds for $\alpha,\beta \in [0,\infty]$ with $ \alpha\beta=1$.
\end{proposition}
The duality relations in \pref{pro:duality} can be used to derive inequalities between different R\'enyi conditional entropies \cite[corollary 4]{TomaBH14} as well as upper and lower bounds for these conditional entropies. 
\begin{lemma}
	\label{lem:RCEineq}
	Suppose $\alpha\in [\frac{1}{2},\infty]$ and $\rho$ is a bipartite state shared by Alice and Bob. Then
	\begin{align}
	H_\alpha^{\downarrow}(A|B)_{\rho}\leq H_\alpha^{\uparrow}(A|B)_{\rho}\leq H_{2-\frac{1}{\alpha}}^{\downarrow}(A|B)_{\rho},\\ \overline{H}_\alpha^{\downarrow}(A|B)_{\rho}\leq \overline{H}_\alpha^{\uparrow}(A|B)_{\rho}\leq \overline{H}_{2-\frac{1}{\alpha}}^{\downarrow}(A|B)_{\rho},\\
	H_\alpha^{\downarrow}(A|B)_{\rho}\leq
	\overline{H}_\alpha^{\downarrow}(A|B)_{\rho}\leq H_{2-\frac{1}{\alpha}}^{\downarrow}(A|B)_{\rho},
	\\
	H_\alpha^{\uparrow}(A|B)_{\rho}\leq 
	\overline{H}_\alpha^{\uparrow}(A|B)_{\rho}\leq H_{2-\frac{1}{\alpha}}^{\uparrow}(A|B)_{\rho}.
	\end{align}	
	The second inequality in each of  the four equations is saturated whenever $\rho$ is pure. 
\end{lemma}
\begin{remark}
	The  inequalities in \lref{lem:RCEineq} were derived in \cite[corollary 4]{TomaBH14}. The first inequalities in the four equations 
	are reproduced from \eref{eq:RREineqSimple}.
	The paper \rcite{TomaBH14} did not discuss the equality conditions.
	The following proof refines the original proof in \rcite{TomaBH14}, so as to show that the second inequalities in the four equations are saturated 
	when $\rho$ is pure.
\end{remark}

\begin{proof}
Suppose $\alpha\in [\frac{1}{2},\infty]$. Let $\sigma$ be a purification of $\rho$ that is shared by A, B, and E. Then $H_\alpha^{\uparrow}(A|B)_{\rho}=H_\alpha^{\uparrow}(A|B)_{\sigma}$, so that
	\begin{align}\label{eq:RCEineqProof}
	H_\alpha^{\uparrow}(A|B)_{\rho}=
	-\overline{H}_\beta^{\downarrow}(A|E)_{\sigma}\leq
	-H_\beta^{\downarrow}(A|E)_{\sigma}=H_\gamma^{\downarrow}(A|B)_\rho
	\end{align}
	according to \pref{pro:duality}, where $\beta=1/\alpha$ and $\gamma=2-\beta=2-(1/\alpha)$. This result confirms the first equation in \lref{lem:RCEineq} given that the first inequality there is trivial. If $\rho$ is pure, then $\sigma_{AE}$ must be a product state, so that the inequality in \eref{eq:RCEineqProof} is saturated according to \eref{eq:RCEproductS}, which implies that $H_\alpha^{\uparrow}(A|B)_{\rho}= H_{2-\frac{1}{\alpha}}^{\downarrow}(A|B)_{\rho}$. The other three equations in \lref{lem:RCEineq} can be derived in a similar manner.
\end{proof} 

\begin{lemma}\label{lem:RCEub}
	\begin{align}
H_\alpha^{\downarrow}(A|B)_{\rho}\leq H_\alpha^{\uparrow}(A|B)_{\rho}&\leq
S_\alpha(\rho_A) \quad \forall \alpha\in [0,\infty], \label{eq:RCEaub}\\  
\overline{H}_\alpha^{\downarrow}(A|B)_{\rho}\leq \overline{H}_\alpha^{\uparrow}(A|B)_{\rho}&\leq S_\alpha(\rho_A)\quad \forall \alpha\in \Bigl[ \frac{1}{2},\infty\Bigr].
\label{eq:RCEbub}	
\end{align}
All the four inequalities are saturated simultaneously for all $\alpha$
iff $\rho$ is a product state.
\end{lemma}
\begin{remark}
The inequality $\overline{H}_\alpha^{\uparrow}(A|B)_{\rho}\leq S_\alpha(\rho_A)$  was derived in \rcite{LediRD17}.
\end{remark}
\begin{proof}
If $\alpha\in\bigl[\frac{1}{2},\infty\bigr]$, then 
\begin{equation}
\overline{H}_\alpha^{\uparrow}(A|B)_{\rho}
=-\min_{\sigma_B} \underline{S}_\alpha(\rho_{AB} \| I_A \otimes \sigma_B)\leq- \underline{S}_\alpha(\rho_{A} \| I_A )=-S_\alpha(\rho_{A} \| I_A )=S_\alpha(\rho_A),
\end{equation}
where the inequality is due to the monotonicity of $\underline{S}_\alpha$ under the partial trace. This observation confirms \eref{eq:RCEbub} given that the first inequality there is obvious. 
By the same token, $H_\alpha^{\uparrow}(A|B)_{\rho}
\leq S_\alpha(\rho_A)$ for $\alpha\in [0,2]$. In addition $H_\alpha^{\uparrow}(A|B)_{\rho}\leq \overline{H}_\alpha^{\uparrow}(A|B)_{\rho}\leq S_\alpha(\rho_A)$ for $\alpha\in\bigl[\frac{1}{2},\infty\bigr]$, which confirms \eref{eq:RCEaub}. 

If $\rho$ is a product state, then the four inequalities in \lref{lem:RCEub} are saturated according to \eref{eq:RCEproductS}.
Conversely, if all the four inequalities are saturated for all $\alpha$, then  $ H(A|B)_{\rho}= S(\rho_A)$, which implies that
$S(\rho_{AB}\|\rho_A \otimes \rho_B)=0$, so that  $\rho_{AB}=\rho_A \otimes \rho_B $ is a product state.
\end{proof}

The following lemma generalizes the Araki-Lieb inequality $H(A|B)\geq -S(\rho_A)$ \cite{ArakL70}, in which \eref{eq:RCEblb} was derived in \rcite{LediRD17}.
\begin{lemma}\label{lem:RCElb}
	\begin{align}
	H_\alpha^{\downarrow}(A|B)_{\rho}&\geq 	-S_{2-\alpha}(\rho_A)\quad \forall \alpha\in [0,2],\\	
	H_\alpha^{\uparrow}(A|B)_{\rho}&\geq -S_{\frac{1}{\alpha}}(\rho_A)\quad \forall \alpha\in [0,2],\\ 
\overline{H}_\alpha^{\downarrow}(A|B)_{\rho}&\geq 	-S_{\frac{1}{\alpha}}(\rho_A) \quad \forall \alpha\in [0,\infty], \\
\overline{H}_\alpha^{\uparrow}(A|B)_{\rho}&\geq 	-S_{\frac{\alpha}{2\alpha-1}}(\rho_A)\quad \forall \alpha\in \Bigl[ \frac{1}{2},\infty\Bigr].\label{eq:RCEblb}
	\end{align}
All the four inequalities are saturated  simultaneously for all $\alpha$
iff 
the system A is independent of the environment of $\rho$.
In particular,
all the four inequalities are saturated when $\rho$ is pure.
\end{lemma}
\begin{remark}
	When $\rho$ is pure, the system A is independent of the environment of $\rho$. However, the converse does not hold in general. 
	For example, when $\rho=\rho_A\otimes \rho_B$ with $\rho_A$ a pure state, the system A is independent of the environment of $\rho$, although $\rho$ is not necessarily pure.
\end{remark}

\begin{proof}
Let $\sigma$ be a purification of $\rho$ that is shared by A, B, and E. If $\alpha\in [0,2]$, then 
\begin{equation}
H_\alpha^{\downarrow}(A|B)_{\rho}=-	H_{2-\alpha}^{\downarrow}(A|E)_{\sigma}\geq -S_{2-\alpha}(\rho_A) 
\end{equation}
according to \pref{pro:duality} and \lref{lem:RCEub}.
If the system A is independent of the environment of $\rho$, that is, 
if $\sigma_{AE}$ is a product state, 
then  the inequality above is saturated according to \lref{lem:RCEub}.
The other three inequalities in \lref{lem:RCElb} can be derived in a similar manner, and they are saturated when $\sigma_{AE}$ is a product state by the same token.

Conversely, if all the four inequalities in \lref{lem:RCElb} are saturated for all $\alpha$, then we have
 $ H(A|E)=S(\rho_A)=S(\sigma_A)$, which implies that $S(\sigma_{AE}\|\sigma_A \otimes \sigma_E)=0$, so that $\sigma_{AE}=\sigma_A \otimes \sigma_E $. Therefore, the system A is independent of the environment of $\rho$.
\end{proof}

\subsection{\label{sec:RREE}Entanglement measures based on R\'enyi relative entropies}

Given a bipartite state $\rho$ shared by Alice and Bob, we can define two types of R\'enyi relative entropies of entanglement as
\begin{align}
E^\scA_{\rmr,\alpha}(\rho):=
\min_{\sigma \in \scA}S_\alpha(\rho\|\sigma),\quad
\underline{E}^\scA_{\rmr,\alpha}(\rho):=
\min_{\sigma \in \scA}\underline{S}_\alpha(\rho\|\sigma),
\end{align}
where $\scA$ may denote one of the three sets, the set  of separable states, that of PPT states, and that of nondistillable states. To simplify the notation, we will drop this superscript if a statement applies to all three choices of $\scA$. Incidentally,  R\'enyi relative entropies are also useful to quantifying quantum correlations \cite{MisrBPS15}. 
\begin{proposition}\label{pro:RREEmono}
$E_{\rmr,\alpha}(\rho)$ for $\alpha\in [0,2]$ and $\underline{E}_{\rmr,\alpha}(\rho)$ for $\alpha\in [\frac{1}{2},\infty]$ do not increase under local operations and classical communication (LOCC). 
\end{proposition}
This proposition shows that $E_{\rmr,\alpha}(\rho)$ for $\alpha\in [0,2]$ and $\underline{E}_{\rmr,\alpha}(\rho)$ for $\alpha\in [\frac{1}{2},\infty]$ are proper entanglement measures. This conclusion follows from the following two facts: First, the R\'enyi relative entropies
${S}_\alpha(\rho\|\sigma)$
for $\alpha 
\in [0,2]$ and 
$\underline{S}_\alpha(\rho\|\sigma)$ for 
$\alpha \in [\frac{1}{2},\infty]$ satisfy  the data-processing inequality \cite{Petz86} \cite[lemma 8.7]{Haya17book}\cite[lemma 3.4]{Haya17book2}; see \eqref{eq:DPI-RREa} and \eqref{eq:DPI-RREb}.  Second, 
the set  of separable states is invariant under LOCC, and so are the set of PPT states and that of nondistillable states. 
 Actually, here LOCC can be replaced by CPTP maps that preserve the set $\scA$ of concern. Outside these parameter ranges, $E_{\rmr,\alpha}(\rho)$ and $\underline{E}_{\rmr,\alpha}(\rho)$ do not satisfy basic requirements for entanglement measures, but  they are still useful in our study.

Incidentally, 
 the quantities 
 $\exp[(\alpha-1)E_{\rmr,\alpha}(\rho)]$ with $\alpha \in (1,2]$
 and $\exp[(\alpha-1)\underline{E}_{\rmr,\alpha}(\rho)]$ with $\alpha \in (1,\infty]$
 are convex in $\rho$ due to the joint convexity of the corresponding R\'enyi relative entropies \eref{eq:H7} \cite[lemma 3.4]{Haya17book2} .
 By contrast, the quantities 
 $\exp[(\alpha-1)E_{\rmr,\alpha}(\rho)]$ with $\alpha \in [0,1)$
 and $\exp[(\alpha-1)\underline{E}_{\rmr,\alpha}(\rho)]$ with $\alpha \in [\frac{1}{2},1)$
 are concave \cite[lemma 3.4]{Haya17book2}.
 Taking the logarithm, we find that
 the entanglement measures 
 $E_{\rmr,\alpha}(\rho)$ with $\alpha \in [0,1)$
 and $\underline{E}_{\rmr,\alpha}(\rho)$ with $\alpha \in [\frac{1}{2},1)$
 are convex in $\rho$.

In the limit $\alpha\to 1$, both R\'enyi relative entropies of entanglement $E^\scA_{\rmr,\alpha}(\rho)$ and $\underline{E}^\scA_{\rmr,\alpha}(\rho)$
approach the conventional relative entropy of entanglement \cite{VedrPRK97,VedrP98,HoroHHH09}
\begin{equation}
E^\scA_\rmr(\rho):=
\min_{\sigma \in \scA}S(\rho\|\sigma).
\label{eq:erab}
\end{equation}
In another limit $\alpha\to \infty$, the variant $\underline{E}^\scA_{\rmr,\alpha}(\rho)$ approaches the logarithmic robustness of entanglement \cite{Datt09,Datt09b,ZhuCH10}
\begin{align}
\underline{E}^\scA_{\rmr,\infty}(\rho)=E^\scA_\lr(\rho):=\ln (1+E^\scA_\caR(\rho)),\label{H5-3BC}
\end{align}
where
\begin{equation}\label{eq:RoE}
E^\scA_\caR(\rho):=\min\left\{x \Big|x\geq0, \; \exists \mbox{ a state } \sigma,\; \frac{\rho+x\sigma}{1+x}\in \scA \right\}
\end{equation}
is the robustness of entanglement (originally called the generalized robustness of entanglement) \cite{HoroHHH09,VidaT99,HarrN03,Stei03,Bran05}. Here $\sigma$ is an arbitrary quantum state, not necessarily contained in $\scA$.
In general, $E_{\rmr,\alpha}(\rho)$ and $\underline{E}_{\rmr,\alpha}(\rho)$ are monotonically increasing with $\alpha$. Therefore, 
\begin{equation}
\underline{E}_{\rmr,\alpha}(\rho)\leq 
E_\lr(\rho) \quad \forall \alpha\in [0,\infty].
\end{equation}
The special case $E_\rmr(\rho)\leq E_\lr(\rho)$ is well known \cite{HayaMMO06,ZhuCH10}.
In addition, the min relative entropy of entanglement $\underline{E}^\scA_{\rmr,1/2}(\rho)$ is equal to a variant of the geometric (measure of) entanglement~\cite{HoroHHH09,WeiG03,StreKB10},
\begin{equation}
\underline{E}^\scA_{\rmr,1/2}(\rho)=E^\scA_\rmG(\rho):=-\ln \max_{\sigma\in \scA} F(\rho,\sigma),
\end{equation}
recall that $\underline{S}_{1/2}(\rho\|\sigma)=-\ln F(\rho,\sigma)$ according to \eref{eq:MinRE}. The measure $E^\scA_\rmG(\rho)$ has a popular variant defined as 
\begin{equation}
\tilde{E}^\scA_\rmG(\rho):=1- \max_{\sigma\in \scA} F(\rho,\sigma).
\end{equation}
In this paper we are more interested in the first variant $E_\rmG(\rho)$ due to its simple connection with R\'enyi relative entropies of entanglement. It is known that $\underline{E}_{\rmr,1/2}(\rho)$ and $E_{\rmr,0}(\rho)$
set upper bounds for the asymptotic exact distillation rate of entanglement \cite[lemma 8.15]{Haya17book}, and both bounds are saturated for  pure states \cite{HayaKMM03,Haya06}\cite[Exercise 8.32]{Haya17book}.

When $\rho$ is a pure state, 
$E_\rmr(\rho)$ is equal to the von Neumann entropy of each reduced state \cite{VedrP98}, while $E_\caR(\rho)$ is equal to the negativity \cite{VidaW02}. Recall that the negativity of a bipartite state $\rho$ is defined as
\begin{equation}
\caN(\rho):=\tr\bigl(\bigl|\rho^{\rmT_A}\bigr|\bigr)-1,
\end{equation}
where $\rmT_A$ denotes the partial transpose with respect to the subsystem A, and $|M|=\sqrt{M^\dagger M}$.
For example, let $\rho=|\psi\>\<\psi|$ with  $|\psi\>=\sum_j \sqrt{\lambda_j} |jj\>$. Then we have
\begin{equation}\label{eq:REERoEpure}
E_\rmr(\rho)=S(\rho_A)=-\sum_j \lambda_j \ln\lambda_j,\quad 
E_\caR(\rho)=\mathcal{N}(\rho)=\bigl(\tr\sqrt{\rho_A}\bigr)^2-1=\biggl(\sum_j \sqrt{\lambda_j}\biggr)^2-1.
\end{equation}

The following lemma provides lower bounds for R\'enyi relative entropies of entanglement in terms of R\'enyi conditional entropies. The special case \eref{eq:REEvsCE} is derived in \cite{PlenVP00}. 
\begin{lemma}\label{lem:EntRenyiCE}
	Any bipartite state $\rho$ on $\caH_A\otimes \caH_B$ satisfies
	\begin{align}
	E_\rmr(\rho)&\geq -H(A|B)_\rho, \label{eq:REEvsCE}\\
	E_{\rmr,\alpha}(\rho)&\geq -H_\alpha^{\uparrow}(A|B)_\rho \quad \forall \alpha\in [0,2], \label{eq:REEvsCE2}\\
	\underline{E}_{\rmr,\alpha}(\rho)&\geq -\overline{H}_\alpha^{\uparrow}(A|B)_\rho\quad\forall  \alpha\in \Bigl[ \frac{1}{2},\infty\Bigr], \label{eq:REEvsCE3}\\
	{E}_\lr
	(\rho)&\geq 
	-\overline{H}_{\infty}^{\uparrow}(A|B)_{\rho}.\label{eq:REEvsCE4}
	\end{align}
\end{lemma}

\begin{proof}
	Let $\sigma$ be an arbitrary nondistillable state. Then $\sigma\leq I_A \otimes \sigma_B$ according to \pref{pro:nondistill} below, so that
	\begin{align}
	& S(\rho\|\sigma)\geq S(\rho\|I_A \otimes \sigma_B),
	\end{align}
	because the relative entropy is monotonically decreasing in the second argument.
	Therefore,
	\begin{align}
	&E^\scA_\rmr(\rho)=\min_{\sigma\in \scA} S(\rho\|\sigma)\geq \min_{\sigma\in \scA}S(\rho\|I_A \otimes \sigma_B)=-H(A|B)_\rho,
	\end{align}
where $\scA$ could be  the set  of separable states, that of PPT states, or that of nondistillable states (note that the first two sets are contained in the third one). 	This observation confirms \eref{eq:REEvsCE}. 
	\Esref{eq:REEvsCE2} and \eqref{eq:REEvsCE3} follow from the same reasoning, note that R\'enyi relative entropies $S_\alpha$ with $\alpha\in [0,2]$ and $\underline{S}_\alpha$ with $\alpha\in \bigl[ \frac{1}{2},\infty\bigr]$ are also monotonically decreasing in the second argument \cite{MullDSF14,WildWY14}\cite[Exercise 5.25]{Haya17book}. \Eref{eq:REEvsCE4} is the limit $\alpha\to \infty$ of~\eref{eq:REEvsCE3}.
\end{proof}

The following proposition was proved in \rcite{HoroH99}. See \rcite{HayaC11} for a partial converse. 
\begin{proposition}[\cite{HoroH99}]\label{pro:nondistill}
Any nondistillable bipartite state $\sigma$ on $\caH_A\otimes \caH_B$ satisfies the reduction criterion, that is, 
\begin{equation}
\sigma\leq I_A \otimes \sigma_B,\quad \sigma\leq \sigma_A \otimes I_B.
\end{equation}
\end{proposition}

\subsection{\label{sec:RREC}Coherence measures based on R\'enyi relative entropies}

Consider a $d$-dimensional Hilbert space $\mathcal{H}$ with a reference basis $\{|j\>\}$. A quantum state $\rho$ is incoherent if it is diagonal with respect to the reference basis. The set of incoherent states is denoted by $\caI$. 
A CPTP map $\Lambda$ is incoherence preserving (also called maximally incoherent) if $\Lambda(\rho)\in \caI$ whenever $\rho\in \caI$. 
Suppose the CPTP map $\Lambda$ has Kraus representation $\{K_j\}$, that is, $\Lambda(\rho)=\sum_j K_j \rho K_j^\dag$ for all $\rho$. Then $\{K_j\}$ is incoherent if each Kraus operator $K_j$ maps every incoherent state to an incoherent state, that is $K_j\rho K_j^\dag/\tr\bigl(K_j\rho K_j^\dag\bigr)\in \caI$ whenever $\rho\in \caI$ \cite{Aber06,BaumCP14,WintY16,StreAP17}.
It is strictly incoherent if in addition $K_j^\dag\rho K_j/\tr\bigl(K_j^\dag\rho K_j\bigr)\in \caI$ whenever $\rho\in \caI$ \cite{WintY16}. A CPTP map is necessarily incoherence preserving if it has an (strictly) incoherent Kraus representation. A pure state of the form $|\psi\>=\sum_j c_j |j\>$ with $|c_j|^2=1/d$ is called maximally coherent because any other state in dimension $d$ can be generated from it under (strictly) incoherent operations \cite{BaumCP14}.

Note that the definition of coherence is basis dependent, and so are many related concepts in the resource theory of coherence, including incoherent states, maximally coherent states, incoherence-preserving operations, and (strictly) incoherent operations. All results about coherence in this paper are stated with respect to a given reference basis. 

\begin{figure}
	\centering
	\includegraphics[width=8cm]{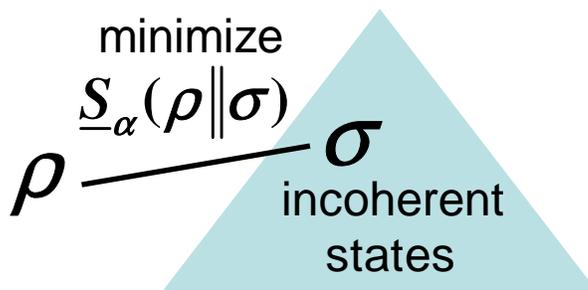}
	\caption{\label{fig:RREC}Illustration of the definition of the R\'enyi relative entropy of coherence $\underline{C}_{\rmr,\alpha}(\rho)$.
		The other  variant ${C}_{\rmr,\alpha}(\rho)$ is defined in a similar way.
	}
\end{figure}

In analogy to entanglement theory, two families of coherence quantifiers can be defined in terms of R\'enyi relative entropies \cite{ChitG16a, ShaoLLX16} as illustrated in \fref{fig:RREC},
\begin{align}
C_{\rmr,\alpha}(\rho):=\min_{\sigma\in \mathcal{I}}S_\alpha(\rho\|\sigma),\quad
\underline{C}_{\rmr,\alpha}(\rho):=\min_{\sigma\in \mathcal{I}}\underline{S}_\alpha(\rho\|\sigma),
\end{align}
where $\mathcal{I}$ denotes the set of incoherent states. Related measures based on Tsallis relative 
entropies were studied in \cite{Rast16}.
Many  results presented in this paper still apply if R\'enyi relative entropies are replaced by 
Tsallis relative 
entropies because the latter are monotonic functions of the former.

\begin{proposition}\label{pro:RRECmono}
	$C_{\rmr,\alpha}(\rho)$ for $\alpha\in [0,2]$ and $\underline{C}_{\rmr,\alpha}(\rho)$ for $\alpha\in [\frac{1}{2},\infty]$ do not increase under incoherence-preserving operations (including incoherent operations). 
\end{proposition}
This proposition shows that $C_{\rmr,\alpha}(\rho)$ for $\alpha\in [0,2]$ and $\underline{C}_{\rmr,\alpha}(\rho)$ for $\alpha\in [\frac{1}{2},\infty]$ are proper coherence measures, in analogy to the corresponding entanglement measures. This conclusion follows from two facts: First, the R\'enyi relative entropies
${S}_\alpha(\rho\|\sigma)$
for $\alpha \in [0,2]$ and 
$\underline{S}_\alpha(\rho\|\sigma)$ for 
$\alpha \in [\frac{1}{2},\infty]$ satisfy  the data-processing inequality \cite{Petz86} \cite[lemma 8.7]{Haya17book}\cite[lemma 3.4]{Haya17book2}; see \eqref{eq:DPI-RREa} and \eqref{eq:DPI-RREb}.  Second, 
the set  of incoherent states is invariant under incoherence-preserving operations. 
Outside these parameter ranges, $C_{\rmr,\alpha}(\rho)$ and $\underline{C}_{\rmr,\alpha}(\rho)$ do not satisfy basic requirements for coherence measures, but  they are still useful in our study.

Incidentally,  
the quantities 
$\exp[(\alpha-1)C_{\rmr,\alpha}(\rho)]$ with $\alpha \in (1,2]$
and $\exp[(\alpha-1)\underline{C}_{\rmr,\alpha}(\rho)]$ with $\alpha \in (1,\infty]$
are convex in $\rho$ due to the joint convexity of the corresponding R\'enyi relative entropies as shown in \eref{eq:H7} \cite[lemma 3.4]{Haya17book2}.
By contrast, the quantities 
$\exp[(\alpha-1)C_{\rmr,\alpha}(\rho)]$ with $\alpha \in [0,1)$
and $\exp[(\alpha-1)\underline{C}_{\rmr,\alpha}(\rho)]$ with $\alpha \in [\frac{1}{2},1)$
are concave \cite[lemma 3.4]{Haya17book2}.
Taking the logarithm, we find that
the coherence measures 
$C_{\rmr,\alpha}(\rho)$ with $\alpha \in [0,1)$
and $\underline{C}_{\rmr,\alpha}(\rho)$ with $\alpha \in [\frac{1}{2},1)$
are convex in $\rho$.

In the limit $\alpha\to 1$, both measures $C_{\rmr,\alpha}(\rho)$ and $\underline{C}_{\rmr,\alpha}(\rho)$ approach the conventional relative entropy of coherence \cite{Aber06,BaumCP14},
\begin{equation}\label{eq:REC}
C_\rmr(\rho):=
\min_{\sigma \in \caI} S(\rho\|\sigma)=S\bigl(\rho^{\diag}\bigr)-S(\rho),
\end{equation}
where $\rho^{\diag}$ is the diagonal part of $\rho$ with respect to the reference basis.
In another limit $\alpha\to \infty$, the measure $\underline{C}_{\rmr,\alpha}(\rho)$ approaches the logarithmic robustness of coherence \cite{ChitG16a},
\begin{align}
\underline{C}_{\rmr,\infty}(\rho)=C_\lr(\rho):=\ln (1+C_\caR(\rho)),\label{H5-3B}
\end{align}
where
\begin{equation}\label{eq:RoC}
C_\caR(\rho):=\min\left\{x \Big|x\geq0, \; \exists \mbox{ a state } \sigma,\; \frac{\rho+x\sigma}{1+x}\in \caI \right\}
\end{equation}
is the robustness of coherence, which is an observable coherence measure and has an operational interpretation in connection with the task of phase discrimination \cite{NapoBCP16,PianCBN16}. 
Similar to $E_{\rmr,\alpha}(\rho)$ and $\underline{E}_{\rmr,\alpha}(\rho)$ discussed in \sref{sec:RREE}, $C_{\rmr,\alpha}(\rho)$ and $\underline{C}_{\rmr,\alpha}(\rho)$ are monotonically increasing with $\alpha$. Therefore, 
\begin{equation}\label{eq:RRECvsRoC}
\underline{C}_{\rmr,\alpha}(\rho)\leq C_\lr(\rho) \quad \forall \alpha\in [0,\infty],
\end{equation}
which implies the inequality $C_\rmr(\rho)\leq C_\lr(\rho)$  derived in \rcite{RanaPWL16}. In addition, the min relative entropy of coherence $\underline{C}_{\rmr,1/2}(\rho)$ is equal to a variant of the geometric (measure of) coherence,
\begin{equation}
\underline{C}_{\rmr,1/2}(\rho)=C_\rmG(\rho):=-\ln \max_{\sigma\in \caI} F(\rho,\sigma),
\end{equation}
which is closely related to another common variant  \cite{StreSDB15},
\begin{equation}
\tilde{C}_\rmG(\rho):=1- \max_{\sigma\in \caI} F(\rho,\sigma).
\end{equation}
In this paper we are more interested in the first variant $C_\rmG(\rho)$ due to its simple connection with R\'enyi relative entropies of coherence. 
As shown in \sref{sec:ExactDistill}, $\underline{C}_{\rmr,1/2}(\rho)$ and $C_{\rmr,0}(\rho)$
set upper bounds for the asymptotic exact distillation rate of coherence, and both bounds are saturated when $\rho$ is pure.

An explicit formula for $C_{\rmr,\alpha}(\rho)$ was derived in \rcite{ChitG16a} as reproduced below. 
\begin{proposition}[\protect{\cite{ChitG16a}}]\label{pro:RRECformulaA}
	\begin{equation}\label{eq:RRECformulaA}
	C_{\rmr,\alpha}(\rho)=\frac{1}{\alpha-1}\ln \bigl\|(\rho^\alpha)^{\diag}\bigr\|_{1/\alpha}\quad \forall \alpha\in [0,\infty],
	\end{equation}
where $(\rho^\alpha)^{\diag}$ denotes the diagonal matrix with the same diagonal as $\rho^\alpha$.
\end{proposition}
\begin{remark}
Note that $C_\rmr(\rho)$ is correctly reproduced in the limit $\alpha\to 1$,
	\begin{equation}\label{eq:RRECformulaAlim}
\lim_{\alpha\to 1} 	C_{\rmr,\alpha}(\rho)=	\lim_{\alpha\to 1}\frac{1}{\alpha-1}\ln \bigl\|(\rho^\alpha)^{\diag}\bigr\|_{1/\alpha}=S(\rho^{\diag})-S(\rho)=C_\rmr(\rho).
	\end{equation}
The paper \cite{ChitG16a} considered $C_{\rmr,\alpha}(\rho)$ only for $\alpha\in [0,2]$, but the formula in \eref{eq:RRECformulaA} is valid for $\alpha\in [0,\infty]$, as demonstrated in the following proof. An alternative proof is presented in the appendix, which is applicable for $\alpha\in [0,2]$.
\end{remark}
\begin{proof}
Suppose $\alpha\geq0$ and $\alpha\neq1$. Then 
\begin{align}
C_{\rmr,\alpha}(\rho)&=\min_{\sigma\in \caI}S_\alpha(\rho\| \sigma)=\min_{\sigma\in \caI}\frac{1}{\alpha-1}\ln\tr(\rho^\alpha\sigma^{1-\alpha})=\min_{\sigma\in \caI}\frac{1}{\alpha-1}\ln\tr\bigl[(\rho^\alpha)^{\diag}\sigma^{1-\alpha}\bigr],
&
\end{align}
where the last equality follows from the assumption that $\sigma$ is diagonal in the reference basis. Let $Q=[(\rho^\alpha)^{\diag}]^{1/\alpha}$ and $\hat{Q}=Q/\tr(Q)$. Then 
\begin{align}
&C_{\rmr,\alpha}(\rho)
=\min_{\sigma\in \caI}\frac{1}{\alpha-1}\ln(Q^\alpha\sigma^{1-\alpha})=\min_{\sigma\in \caI}\frac{1}{\alpha-1}\bigl[\ln(\tr Q)^\alpha+\ln\tr\bigl(\hat{Q}^\alpha\sigma^{1-\alpha}\bigr)\bigr],\nonumber\\
&=\frac{1}{\alpha-1}\ln(\tr Q)^\alpha+\min_{\sigma\in \caI}S_\alpha(\hat{Q}\|\sigma)=\frac{1}{\alpha-1}\ln(\tr Q)^\alpha =\frac{1}{\alpha-1}\ln\bigl\|(\rho^\alpha)^{\diag}\bigr\|_{1/\alpha},
\end{align}
where the minimum is attained when $\sigma=\hat{Q}$. 
\end{proof}

In the case $\alpha=2$, \eref{eq:RRECformulaA} reduces to 
\begin{align}\label{eq:RRECr2}
C_{\rmr,2}(\rho) =\ln\Biggl[\sum_j \biggl(\sum_k |\rho_{jk}|^2\biggr)^{1/2}\Biggr]^2.
\end{align}
In the limit $\alpha\to 0$, \eref{eq:RRECformulaA} yields
\begin{equation}\label{eq:RRECformulaA0}
C_{\rmr,0}(\rho)=-\ln \bigl\|(\Pi_\rho)^{\diag}\bigr\|,
\end{equation}
where $\Pi_\rho$ is the projector onto the support of $\rho$, and $\|M\|=\|M\|_\infty$ denotes the operator norm of $M$. 

When $\rho$ is pure, the formulas of $C_{\rmr,\alpha}(\rho)$ and $\underline{C}_{\rmr,\alpha}(\rho)$ are  derived in \rcite{ChitG16a}.
\begin{proposition}[\protect{\cite{ChitG16a}}]\label{pro:RRECpure}
	Suppose $\rho=|\phi \>\<\phi|$ is a pure state with $|\phi \rangle=\sum_{i}a_i|i\rangle$ and  $|a_i|^2=p_i$. Then we have
	\begin{align}
	C_{\rmr,\alpha}(\rho) 
	&=\begin{cases}
	\frac{\alpha}{\alpha-1}\ln \Bigl(\sum_i p_i^{\frac{1}{\alpha}}\Bigr) 
	& \mbox{ if } \alpha > 0,\\
	-\ln \max_i p_i & \hbox{ if } \alpha = 0,
	\end{cases} 	\label{eq:H4C}\\
	\underline{C}_{\rmr,\alpha}(\rho) 
	&=\begin{cases}
	\frac{2\alpha-1}{\alpha-1}
	\ln \Bigl(\sum_i p_i^{\frac{\alpha}{2\alpha-1}}\Bigr) & \mbox{ if } \alpha > \frac{1}{2},\\
	-\ln \max_i p_i 
	& \mbox{ if } \alpha = \frac{1}{2}.
	\end{cases} \label{eq:H5C}
	\end{align}
\end{proposition}
In the case $\alpha=1$, the formulas in \pref{pro:RRECpure} are understood as proper limits.
Alternatively, these formulas can be expressed as follows,
\begin{align}
C_{\rmr,\alpha}(\rho)&=\begin{cases}
S_{\frac{1}{\alpha}}\bigl(\rho^{\diag}\bigr)=\underline{S}_{\rmr,\alpha}(\rho\|\rho^{\diag})& \forall \alpha\in [0,\infty],\\
S_{2-\frac{1}{\alpha}}(\rho\|\rho^{\diag})&\forall \alpha\in \bigl[ \frac{1}{2},\infty\bigr],
\end{cases} \label{eq:RRECApureAlt}\\
\underline{C}_{\rmr,\alpha}(\rho) &=S_{\frac{\alpha}{2\alpha-1}}\bigl(\rho^{\diag}\bigr)=\underline{S}_{2-\frac{1}{\alpha}}(\rho\|\rho^{\diag})\quad \forall \alpha\in \Bigl[ \frac{1}{2},\infty\Bigr]. \label{eq:RRECBpureAlt}
\end{align}
The reasons
 behind these equalities are explained in  \thsref{thm:RRECub} and
 \ref{thm:RRECabBound} in \sref{sec:bounds} and \thref{thm:RRECabOrder} in \sref{sec:relationRREC}. 

\Pref{pro:RRECpure} implies that any pure state $\rho$ satisfies
\begin{align}
C_\lr(\rho)=\underline{C}_{\rmr,\infty}(\rho)
=C_{\rmr,2}(\rho)=C_{\rmL}(\rho)
=2\ln \bigl(\tr \sqrt{\rho^{\diag}}\bigr).\label{eq:RoCCr2}
\end{align}
Here $C_{\rmL}(\rho):=\ln(1+C_{l_1}(\rho))$ and
\begin{equation}
C_{l_1}(\rho):=\sum_{j\neq k}|\rho_{jk}|
\end{equation}
is the $l_1$-norm of coherence \cite{BaumCP14}, which may be seen as the analogue of the negativity in entanglement theory \cite{RanaPWL16,ZhuHC17}. In particular, the $l_1$-norm of coherence can be uniquely characterized by a few simple axioms in a similar way to the negativity. In addition,  the $l_1$-norm of coherence is equal to the maximum entanglement,  quantified by the negativity, produced by  incoherent operations acting on the system and an incoherent ancilla \cite{ZhuHC17}.

According to theorem~4 in \rcite{PianCBN16}, any state $\rho$ in dimension $d$ satisfies
\begin{equation}\label{eq:RoCL}
\frac{ C_{l_1}(\rho)}{d-1}\leq C_{\caR}(\rho)\leq C_{l_1}(\rho),
\end{equation}
which implies that 
\begin{equation}\label{eq:RoCLL}
\ln\Bigl[\frac{ C_{l_1}(\rho)}{d-1}+1\Bigr]\leq C_\lr(\rho)\leq C_{\rmL}(\rho).
\end{equation}
In conjunction with \eref{eq:RRECvsRoC}, we deduce that
\begin{equation}
\underline{C}_{\rmr,\alpha}(\rho)\leq C_{\rmL}(\rho)\quad \forall \alpha\in [0,\infty],
\end{equation}
which implies that $C_\rmr(\rho)\leq C_{\rmL}(\rho)$ in particular. On the other hand, the lower bound for $C_\lr(\rho)$ in \eref{eq:RoCLL} in general does not apply to 
$C_\rmr(\rho)$. For example, when $d=2$, the upper and lower bounds in \eref{eq:RoCLL} coincide, which implies that $C_\lr(\rho)=C_\rmL(\rho)$. However, the inequality $C_\rmr(\rho)<C_\lr(\rho)=C_\rmL(\rho)$ is strict except when $\rho$ is either maximally coherent or incoherent (cf.~\thref{thm:EntropyRobustEqual} in \sref{sec:relationRREC}). 
 To be concrete, consider
$\rho=\proj{\ps}$ with $\ket{\ps}=\cos\t\ket{0}+\sin\t\ket{1}$. 
We have
\begin{equation}
C_\lr(\rho)=C_\rmL(\rho)=\ln[1+\sin(2\theta)],\quad C_\rmr(\rho)=-\cos^2\theta\ln\cos^2\theta-\sin^2\theta\ln\sin^2\theta.
\end{equation}
It is easy to verify that $C_\rmr(\rho)<C_\rmL(\rho)$ when $\t>0$ is sufficiently small.

The coherence measures introduced above can be generalized to a bipartite or multipartite system, in which case the reference basis is the tensor product of reference bases for respective subsystems. 
The following lemma clarifies the relations between entanglement measures and coherence measures based on R\'enyi relative entropies for a bipartite system. It is an immediate consequence of the definitions and the fact that incoherent states  are separable,  PPT, and nondistillable. The same conclusion also applies to a multipartite system.
\begin{lemma}\label{lem:EntCohInEq}
	\begin{align}
	\label{eq:ec1}
	E_\rmr(\rho) &\le C_\rmr(\rho), \quad & E_\lr(\rho) &\le C_\lr(\rho),
	\\
	E_{\rmr,\alpha}(\rho)
	&\le
	C_{\rmr,\alpha}(\rho),
	\quad &
	\underline{E}_{\rmr,\alpha}(\rho)
	& \le
	\underline{C}_{\rmr,\alpha}(\rho)\quad \forall \alpha\in [0,\infty].
	\end{align}
\end{lemma}
Although coherence measures depend on the choice of local bases (unlike entanglement measures), \lref{lem:EntCohInEq} is applicable to any given choice of local bases.
In  \thref{thm:EntCohMC} in the next section, we will show that all the inequalities in \lref{lem:EntCohInEq} are saturated when $\rho$ is a maximally correlated state \cite{Rain99} as long as the corresponding R\'enyi relative entropies satisfy the data processing inequality. Recall that a \emph{maximally correlated state} has the form  \cite{Rain99}
\begin{equation}
\rho_{\mrm{MC}}:=\sum_{jk}\rho_{jk}\outer{jj}{kk}.
\end{equation}

%%%%%%%%%%%%%%%%%%%%%%%%%%%%%%%%%%%%%%%%%%%%%%%%%%%%%%%%%%%%%%%%%%%%%%%%%%%%%%%%%%%%%%%%%%%%%%%%%%%%%%%%%%%%%%%%%%%%%%%%%%%%%%%%%%%%%%%%%%%%%%%%%%%%%%%%%%%%%%%%%%%%%%%%%%%%%%%%%%%%%%

\section{\label{sec:Connect}Connecting entanglement measures and coherence measures}
In this section we establish an operational one-to-one mapping between entanglement  measures and coherence measures based on R\'enyi relative entropies. To achieve this goal, we first clarify the relations between 
these measures and R\'enyi conditional  entropies for  maximally correlated states. As applications, we derive several analytical formulas for R\'enyi relative entropies of entanglement of maximally correlated states, which reproduce  a number of   classic results on the
relative entropy of entanglement and logarithmic  robustness of entanglement as special cases. 
In addition, the results presented here play crucial roles in understanding several  topics discussed in the following sections, including the additivity of R\'enyi relative entropies of coherence (\sref{sec:Additivity}) and the exact coherence distillation rate (\sref{sec:ExactDistill}).

Our study is inspired by a recent work of Streltsov et al. \rcite{StreSDB15}, which provides a general framework for constructing coherence measures from entanglement measures; see also \rscite{ZhuMCF17,ZhuHC17}. 
Let $\rho$ be any density matrix on $\mathcal{H}_A$ of dimension $d_A$. If $\rho$ is coherent, then it can generate entanglement under incoherent operations acting on the system $\mathcal{H}_A$ and an incoherent density matrix on an ancilla $\mathcal{H}_B$. Given any entanglement measure~$E$, the maximum entanglement generated in this way defines a coherence measure $C_E$ according to \rcite{StreSDB15}. More precisely,
\begin{equation}\label{eq:EntGen}
C_E(\rho):=\lim_{d_B\rightarrow \infty}\left\{\sup_{\Lambda_\rmi}E\left(\Lambda_\rmi\left[\rho\otimes \outer{0}{0}\right]\right)\right\},
\end{equation}
where $d_B$ is the dimension of the ancilla, and the supremum is taken over all incoherent operations $\Lambda_\rmi$. Interestingly, \eref{eq:EntGen} maps the relative entropy of entanglement, geometric entanglement, and negativity to the relative entropy of coherence, geometric coherence, and $l_1$-norm of coherence, respectively, that is, 
$C_E=C_\rmr, C_\rmG, C_{l_1}$ when 
$E=E_\rmr, E_\rmG, \caN$ \cite{StreSDB15,ZhuHC17}. Moreover, it enables  establishing a one-to-one mapping between entanglement measures and coherence measures that are based on the convex roof~\cite{ZhuMCF17}.
Surprisingly, the generalized CNOT gate $\cnot$ is the common optimal incoherent operation with respect to all these entanglement measures. Recall that $\cnot$ corresponds to conjugation by the unitary $U_{\mrm{CNOT}}$ defined as follows,
\begin{equation}
U_{\mrm{CNOT}}|jk\>=\begin{cases}
|j (j+k)\>
& k<d_A,\\
|j k\> &k\geq d_A,
\end{cases}
\end{equation}
where the addition is modulo $d_A$, assuming that $d_B\geq d_A$. The operation $\cnot$ turns any state
$\rho=\sum_{jk}\rho_{jk}\outer{j}{k}$ on $\caH_A$ into a maximally correlated state on $\caH_A\otimes \caH_B$ \cite{Rain99, StreSDB15,WintY16},
\begin{equation}
\rho_{\mrm{MC}}=\cnot\left[\rho\otimes \outer{0}{0}\right]=\sum_{jk}\rho_{jk}\outer{jj}{kk}.
\end{equation}
It is worth mentioning that any bipartite entangled pure state is equivalent to a maximally correlated state under local unitary transformations.

Here we shall extend the operational connection between entanglement and coherence to measures based on R\'enyi relative entropies.
By virtue of \eref{eq:EntGen}, we can define 
two families of coherence quantifiers based on the two families of 
R\'enyi relative entropies of entanglement as illustrated in  \fref{fig:EntCoh}, 
\begin{align}
C_{E_{\rmr,\alpha}}(\rho)&:=\lim_{d_B\rightarrow \infty}\left\{\sup_{\Lambda_\rmi}
E_{\rmr,\alpha}\left(\Lambda_\rmi\left[\rho\otimes \outer{0}{0}\right]\right)\right\}, \label{eq:EntCohRenyiA}\\
C_{\underline{E}_{\rmr,\alpha}}(\rho)&:=\lim_{d_B\rightarrow
	\infty}
\left\{\sup_{\Lambda_\rmi}
\underline{E}_{\rmr,\alpha}\left(\Lambda_\rmi\left[\rho\otimes \outer{0}{0}\right]
\right)\right\}. \label{eq:EntCohRenyiB}
\end{align}
According to \pref{pro:RREEmono},  $E_{\rmr,\alpha}(\rho)$ for $\alpha\in [0,2]$ and $\underline{E}_{\rmr,\alpha}(\rho)$ for $\alpha\in [\frac{1}{2},\infty]$ are proper entanglement measures. Therefore, $C_{E_{\rmr,\alpha}}(\rho)$ for $\alpha\in [0,2]$ and $C_{\underline{E}_{\rmr,\alpha}}(\rho)$ for $\alpha\in [\frac{1}{2},\infty]$ are proper coherence measures.  In the limit $\alpha\rightarrow 1$, both $E_{\rmr,\alpha}$ and $\underline{E}_{\rmr,\alpha}$ approach the relative entropy of entanglement $E_\rmr$, so $C_{E_{\rmr,\alpha}}$ and $C_{\underline{E}_{\rmr,\alpha}}$ reduce to $C_{E_{\rmr}}$, which is equal to the relative entropy of coherence $C_\rmr$ according to \rcite{StreSDB15}. In another limit $\alpha\rightarrow \infty$, $\underline{E}_{\rmr,\alpha}$ approaches the logarithmic robustness of entanglement $E_\lr$, and \eref{eq:EntCohRenyiB} takes on the form
\begin{equation}\label{eq:EntGenRo}
C_{E_\lr}(\rho):=\lim_{d_B\rightarrow \infty}\left\{
\sup_{\Lambda_\rmi}E_\lr
\left(\Lambda_\rmi\left[\rho\otimes \outer{0}{0}\right]\right)\right\}.
\end{equation}

\begin{figure}
	\centering
	\includegraphics[width=14cm]{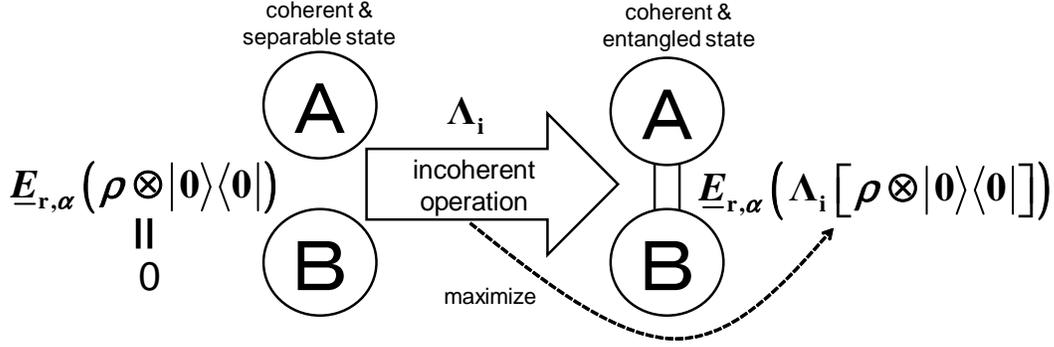}
	\caption{\label{fig:EntCoh}Illustration of the definition of the coherence measure $C_{\underline{E}_{\rmr,\alpha}}(\rho)$ as the maximum R\'enyi relative entropy of entanglement $\underline{E}_{\rmr,\alpha}$ generated by incoherent operations acting on the system and an incoherent ancilla.}
\end{figure}

To achieve our goal, we first show that 
the inequalities between R\'enyi relative entropies of entanglement and R\'enyi conditional entropies as well as R\'enyi relative entropies of coherence in  \lsref{lem:EntRenyiCE} and \ref{lem:EntCohInEq} are saturated for maximally correlated states (for the parameter ranges of interest).
\begin{theorem}\label{thm:EntCohMC}
	Any maximally correlated state $\rho_{\mrm{MC}}$ satisfies the following relations,
	\begin{align}
	E_\rmr({\rho_{\mrm{MC}}})&= C_{\rmr}(\rho_{\mrm{MC}})= -H(A|B)_{\rho_{\mrm{MC}}},\label{eq:EntCohMC1}\\
	E_{\rmr,\alpha}(\rho_{\mrm{MC}})&= C_{\rmr,\alpha}(\rho_{\mrm{MC}})=-H_\alpha^{\uparrow}(A|B)_{\rho_{\mrm{MC}}}\quad \forall \alpha \in [0,2], \label{eq:EntCohMCa}\\
	\underline{E}_{\rmr,\alpha}(\rho_{\mrm{MC}})&= \underline{C}_{\rmr,\alpha}(\rho_{\mrm{MC}})
	=-\overline{H}_\alpha^{\uparrow}(A|B)_{\rho_{\mrm{MC}}}\quad \forall \alpha \in \Bigl[\frac{1}{2},\infty\Bigr],\label{eq:EntCohMCb}\\
	{E}_\lr
	(\rho_{\mrm{MC}})&= {C}_\lr
	(\rho_{\mrm{MC}})
	=-\overline{H}_{\infty}^{\uparrow}(A|B)_{\rho_{\mrm{MC}}}.\label{eq:EntCohMCro}
	\end{align} 
\end{theorem}
\begin{remark}
Although coherence measures depend on the choice of local bases (unlike entanglement measures), \thref{thm:EntCohMC} is applicable to any given choice of local bases. In addition, 
this theorem applies to entanglement measures defined with respect to three type of states, namely, separable states, PPT states, and nondistillable states (see \sref{sec:RREE}). Similar remarks apply to many other results presented in this paper. 
\end{remark}

\begin{proof}
	Let $P$ be the projector onto the space spanned by $|jj\>$ for all $j$ and define the CPTP map $\Lambda_P$ by $\Lambda_P(\rho): =P\rho P+(1-P)\rho(1-P)$. Then $\Lambda_P(\rho_{\mrm{MC}}) =\rho_{\mrm{MC}}$, so that 
	\begin{align}
	S(\rho_{\mrm{MC}}\|I \otimes \sigma_B)
	\geq S\bigl(\Lambda_P(\rho_{\mrm{MC}})\|\Lambda_P(I \otimes \sigma_B)\bigr)=S\bigl(\rho_{\mrm{MC}}\|P(I \otimes \sigma_B)P\bigr)\label{eq:REmapMC}
	\end{align}
	for any state $\sigma_B$ on $\mathcal{H}_B$. Observing that $P(I \otimes \sigma_B)P=\sum_j (\sigma_B)_{jj}|jj\>\<jj|$ is a normalized incoherent state, we conclude that 
	\begin{align}
	-H(A|B)_{\rho_{\mrm{MC}}}=\min_{\sigma_B}
	S(\rho_{\mrm{MC}}\|I \otimes \sigma_B)\geq \min_{\sigma\in \caI}S(\rho_{\mrm{MC}}\|\sigma)=C_\rmr(\rho_{\mrm{MC}}).
	\end{align}
This result confirms \eref{eq:EntCohMC1} given the inequality $-H(A|B)_{\rho_{\mrm{MC}}}\leq E_\rmr(\rho_{\mrm{MC}})\leq C_{\rmr}(\rho_{\mrm{MC}})$ according to \lsref{lem:EntRenyiCE} and \ref{lem:EntCohInEq}. 
	
	\Esref{eq:EntCohMCa}, \eqref{eq:EntCohMCb}, and \eqref{eq:EntCohMCro} can be proved in a similar way. In particular, \eref{eq:REmapMC}
	still applies if $S$ is replaced by $S_\alpha$ with $\alpha\in [0,2]$ or $\underline{S}_\alpha$ with $\alpha\in \bigl[\frac{1}{2},\infty\bigr]$. Therefore, $-H_\alpha^{\uparrow}(A|B)_{\rho_{\mrm{MC}}}\geq C_{\rmr,\alpha}(\rho_{\mrm{MC}})$ and $-\overline{H}_\alpha^{\uparrow}(A|B)_{\rho_{\mrm{MC}}}\geq \underline{C}_{\rmr,\alpha}(\rho_{\mrm{MC}})$ for the given parameter ranges, which imply \esref{eq:EntCohMCa} and \eqref{eq:EntCohMCb} in view of \lsref{lem:EntRenyiCE} and \ref{lem:EntCohInEq}. Finally, \eref{eq:EntCohMCro} is derived from \eqref{eq:EntCohMCb} by taking the limit $\alpha\rightarrow \infty$. 
\end{proof}

\begin{remark}
	The equality $E_{\rmr}(\rho_{\mrm{MC}})=-H(A|B)_{\rho_{\mrm{MC}}}$ in \eref{eq:EntCohMC1}
	is known before \cite{PlenVP00} \cite[(8.143)]{Haya17book}; also, the equality
	$E_{\rmr}(\rho_{\mrm{MC}})=C_{\rmr}(\rho_{\mrm{MC}})$
	has been derived in \rscite{StreSDB15}. In addition, the relations
	$	E_{\rmr,\alpha}(\rho_{\mrm{MC}})=-H_\alpha(A|B)_{\rho_{\mrm{MC}}}$
	and $\underline{E}_{\rmr,\alpha}(\rho_{\mrm{MC}})
	=-\overline{H}_\alpha(A|B)_{\rho_{\mrm{MC}}}$ in \esref{eq:EntCohMCa} and \eqref{eq:EntCohMCb}
	were stated in \cite[lemma 8.9]{Haya17book}.
	However, the derivation there contains an error, which is fixed here.
\end{remark}

Now we can establish an operational connection between entanglement measures and coherence measures based on R\'enyi relative entropies.
\begin{theorem}\label{thm:CohTwoMeasCoincide}
	We have the following relations,
	\begin{align}
	C_{E_\rmr}(\rho)&= C_{\rmr}(\rho)\label{eq:H22a},\\
	C_{E_{\rmr,\alpha}}(\rho)&= C_{\rmr,\alpha}(\rho) \quad \forall \alpha\in [0,2], \label{eq:H22}\\
	C_{\underline{E}_{\rmr,\alpha}}(\rho)&= \underline{C}_{\rmr,\alpha}(\rho)\quad \forall  \alpha\in \Bigl[ \frac{1}{2},\infty\Bigr],
	\label{eq:H32}\\
	C_{{E}_\lr}(\rho)&= {C}_\lr(\rho) .
	\label{eq:H32+}
	\end{align}
\end{theorem}
\begin{remark}
	According to the following proof, the generalized CNOT gate is the common optimal incoherent operation that achieves	 the supremums in the definitions of $C_{E_\rmr}(\rho)$,
	$C_{E_{\rmr,\alpha}}(\rho)$, $C_{\underline{E}_{\rmr,\alpha}}(\rho)$, and 
	$C_{E_\lr}(\rho)$. Here the special case \eref{eq:H22a} was derived in \rcite{StreSDB15}. 
	In view of the relation $\underline{S}_{1/2}(\rho\|\sigma)=-\ln F(\rho,\sigma)$, \thref{thm:CohTwoMeasCoincide} implies  that $C_{E_\rmG}(\rho)=C_\rmG(\rho)$ and
	$C_{\tilde{E}_\rmG}(\rho)=\tilde{C}_\rmG(\rho)$, 	which were derived in \rscite{StreSDB15,ZhuMCF17} based on  different approaches. 
\end{remark}

\begin{proof}Let $\Lambda_\rmi$ be an arbitrary incoherence-preserving operation acting on the system and the ancilla. Then
	\begin{align}
	&E_\rmr\left(\Lambda_\rmi\left[\rho\otimes \outer{0}{0}\right]\right)\leq C_\rmr\left(\Lambda_\rmi\left[\rho\otimes \outer{0}{0}\right]\right) \leq C_\rmr\left(\rho\otimes \outer{0}{0}\right) 
	=C_\rmr(\rho)\label{eq:H2631}
	\end{align}
	according to \lref{lem:EntCohInEq} and \pref{pro:RRECmono}. 
	By \thref{thm:EntCohMC}, the two inequalities are saturated when $\Lambda_\rmi$ is the generalized CNOT gate, in which case $\Lambda_\rmi\left[\rho\otimes \outer{0}{0}\right]$ is maximally correlated. 
	This observation confirms \eref{eq:H22a}. 
	
	\Esref{eq:H22}, \eqref{eq:H32}, and \eqref{eq:H32+} follow from the same reasoning as above, note that \eref{eq:H2631} still holds if $E_\rmr$ is replaced by $E_{\rmr,\alpha}$, $\underline{E}_{\rmr,\alpha}$, and $E_\lr$, while $C_\rmr$ is replaced by $C_{\rmr,\alpha}$, $\underline{C}_{\rmr,\alpha}$, and $C_\lr$ accordingly. 
\end{proof}

\Thref{thm:EntCohMC} is useful not only in connecting entanglement measures and coherence measures based on R\'enyi relative entropies, but also in studying entanglement measures of maximally correlated states, including bipartite pure states.	 
\begin{corollary}\label{cor:RREEa}
	Suppose $\rho$ is a maximally correlated state. Then 
	\begin{equation}
	E_{\rmr,\alpha}(\rho)=\frac{1}{\alpha-1}\ln \bigl\|(\rho^\alpha)^{\diag}\bigr\|_{1/\alpha}\quad \forall \alpha\in [0,2].
	\end{equation}
\end{corollary}
This corollary is a consequence of \thref{thm:EntCohMC} and  \pref{pro:RRECformulaA}. In conjunction with \eref{eq:RRECformulaAlim}, we deduce that 
\begin{equation}
\lim_{\alpha\to 1}E_{\rmr,\alpha}(\rho)
=S(\rho^{\diag})-S(\rho)=S(\rho_A)-S(\rho)=E_\rmr(\rho),
\end{equation}
which reproduces the relative entropy of entanglement of maximally correlated states \cite{PlenVP00} \cite[(8.143)]{Haya17book}, including bipartite pure states \cite{VedrP98}.

\begin{corollary}\label{cor:RREEpure}
	Suppose $\rho=|\phi\>\<\phi|$ is a bipartite pure state with $|\phi \rangle=\sum_{i}a_i|ii\rangle$ and $|a_i|^2=p_i$. Then 
	\begin{align}
	E_{\rmr,\alpha}(\rho) 
	&=\begin{cases}
	\frac{\alpha}{\alpha-1}\ln \Bigl(\sum_i p_i^{\frac{1}{\alpha}}\Bigr) 
	& \mbox{ if } \alpha > 0,\\
	-\ln \max_i p_i & \hbox{ if } \alpha = 0,
	\end{cases} 	\label{eq:H4}\\
	\underline{E}_{\rmr,\alpha}(\rho) 
	&=\begin{cases}
	\frac{2\alpha-1}{\alpha-1}
	\ln \Bigl(\sum_i p_i^{\frac{\alpha}{2\alpha-1}}\Bigr) & \mbox{ if } \alpha > \frac{1}{2},\\
	-\ln \max_i p_i 
	& \mbox{ if } \alpha = \frac{1}{2}.
	\end{cases} \label{eq:H5}
	\end{align}
\end{corollary}
\Crref{cor:RREEpure} is a consequence of \thref{thm:EntCohMC} and \pref{pro:RRECpure}. It reproduces the relative entropy of entanglement of bipartite pure states \cite{VedrP98} in the limit $\alpha\to 1$. 
In addition,  it reproduces the logarithmic robustness of entanglement in another limit $\alpha\to\infty$ \cite{VidaT99,HarrN03,Stei03} and  implies that any bipartite pure state $\rho$ on $\caH_A\otimes \caH_B$ satisfies
\begin{align}
E_\lr(\rho)&=\underline{E}_{\rmr,\infty}(\rho)
=E_{\rmr,2}(\rho)=\caN_\rmL(\rho)=2\ln \bigl(\tr\sqrt{\rho_A}\bigr),
\end{align}
where $\caN_\rmL(\rho):=\ln(1+\caN(\rho))$ is the logarithmic negativity \cite{HoroHHH09,VidaW02}.

\begin{corollary}\label{cor:RoENegMC}
	If $\rho$ is a $d\times d$ maximally correlated state, then 
	\begin{equation}\label{eq:RobustNegM}
	\frac{ \caN(\rho)}{d-1}\leq E_\caR(\rho)\leq \caN(\rho).
	\end{equation}
\end{corollary}
\begin{proof}
	If $\rho$ is a $d\times d$ maximally correlated state, then $\rho$ is supported on a $d$-dimensional subspace spanned by $d$ computational-basis states. Therefore, $\frac{ C_{l_1}(\rho)}{d-1}\leq C_{\caR}(\rho)\leq C_{l_1}(\rho)$
according to  theorem~4 in \rcite{PianCBN16}; cf.~\eref{eq:RoCL} in \sref{sec:RREC}. Now the corollary follows from the equality
 $E_\caR(\rho)=C_\caR(\rho)$  presented in \thref{thm:EntCohMC} and the equality $\caN(\rho)=C_{l_1}(\rho)$ \cite{RanaPWL16,ZhuHC17},  which is  straightforward to verify.
\end{proof}
\Crref{cor:RoENegMC} above implies that $E_\caR(\rho)= \caN(\rho)$ if $\rho$ is a two-qubit maximally correlated state. This requirement is sufficient but not necessary. Indeed, the equality $E_\caR(\rho)= \caN(\rho)$ holds for all Bell-diagonal states, not all of which are maximally correlated.
To see this, consider the Bell-diagonal state $\rho_{\mrm{BD}}(p):=\sum^3_{j=0}p_j \proj{\Psi_j}$, where $p = (p_0, p_1, p_2, p_3)$ is a probability distribution with $p_0\ge 1/2$ and 
\begin{equation}
\begin{aligned}
\ket{\Psi_0}&=\frac{1}{\sqrt2}(\ket{00}+\ket{11}),\quad
&\ket{\Psi_1}&=\frac{1}{\sqrt2}(\ket{00}-\ket{11}),\\
\ket{\Psi_2}&=\frac{1}{\sqrt2}(\ket{01}+\ket{10}),\quad
&\ket{\Psi_3}&=\frac{1}{\sqrt2}(\ket{01}-\ket{10})
\end{aligned}
\end{equation}
are four Bell states, which form a Bell basis.
The Bell-diagonal state $\rho_{\mrm{BD}}(p)$ is maximally correlated iff $p_2=p_3=0$.
Calculation shows that 
\begin{equation}
E_\caR(\rho_{\mrm{BD}}(p))=\mathcal{N}(\rho_{\mrm{BD}}(p))=2p_0-1, \label{Q1}
\end{equation}
where the equality $E_\caR(\rho_{\mrm{BD}}(p))=2p_0-1$
follows from \cite[(29)]{ZhuCH10}.

\section{\label{sec:Additivity}Additivity of R\'enyi relative entropies of coherence}
In quantum information processing, it is often more efficient to process a family of quantum states collectively. In this context, it is natural to ask whether the resource content of this family is equal to the sum of the resource contents of individual members. Additive resource measures are particularly appealing because they can significantly simplify the task of quantifying resources.  By virtue  of \thref{thm:EntCohMC},
in this section we prove that all R\'enyi relative entropies of coherence defined in \sref{sec:RREC} are additive, as long as they are monotonic under incoherence-preserving operations. Accordingly, R\'enyi relative entropies of entanglement defined in \sref{sec:RREE} are additive for maximally correlated states, although they are  not additive in general~\cite{ZhuCH10}.

To achieve our goal, we first recall the additivity properties of 
R\'{e}nyi conditional entropies, which can be proved using the duality relations presented in \pref{pro:duality}.
\begin{proposition}[\protect{\cite[lemma 7]{HayaT16}}]\label{pro:AdditivityRCE}
	Any pair of states $\rho_1$ and $\rho_2$ shared by Alice and Bob satisfies the following additivity relations:
	\begin{align}
	{H}_\alpha^{\uparrow}(A|B)_{\rho_1\otimes \rho_2}
	&={H}_\alpha^{\uparrow}(A|B)_{\rho_1}
	+
	{H}_\alpha^{\uparrow}(A|B)_{\rho_2}\quad \forall \alpha \in [0,\infty],
	\label{eq:AddRREa}\\
	\overline{H}_\alpha^{\uparrow}(A|B)_{\rho_1\otimes \rho_2}
	&=\overline{H}_\alpha^{\uparrow}(A|B)_{\rho_1}
	+
	\overline{H}_\alpha^{\uparrow}(A|B)_{\rho_2}\quad \forall \alpha \in\Bigl[\frac{1}{2},\infty\Bigr].
	\label{eq:AddRREb}
	\end{align}
\end{proposition}
Note that the other two types of  conditional entropies ${H}_\alpha^{\downarrow}(A|B)_{\rho}$ and $\overline{H}_\alpha^{\downarrow}(A|B)_{\rho}$ are obviously additive. Combining  \thref{thm:EntCohMC} with  \pref{pro:AdditivityRCE},
we can prove the additivity of R\'enyi relative entropies of coherence, including the logarithmic robustness of coherence.
\begin{theorem}\label{thm:AdditivityRREC}
	\begin{align}
	C_{\rmr}(\rho_1\otimes \rho_2) &=
	C_{\rmr}(\rho_1) +C_{\rmr}(\rho_2), \label{eq:AdditivityREC}\\
	C_{\rmr,\alpha}(\rho_1\otimes \rho_2) &=
	C_{\rmr,\alpha}(\rho_1) +C_{\rmr,\alpha}(\rho_2) \label{eq:AdditivityRREC} \quad\forall \alpha\in [0,\infty] ,\\
	\underline{C}_{\rmr,\alpha}(\rho_1\otimes \rho_2)&=
	\underline{C}_{\rmr,\alpha}(\rho_1)
	+\underline{C}_{\rmr,\alpha}(\rho_2) \quad\forall \alpha\in \Bigl[ \frac{1}{2},\infty\Bigr], \label{eq:AdditivityRRECb} \\
	C_\lr(\rho_1\otimes \rho_2)&=
	C_\lr(\rho_1)
	+ C_\lr(\rho_2). \label{eq:AdditivityRoC}
	\end{align}
\end{theorem}
\Thref{thm:AdditivityRREC} is of fundamental interest to understanding the resource theory of coherence and its distinction from the resource theory of entanglement. Recall that most entanglement measures are in general not additive. In addition, \thref{thm:AdditivityRREC}
can significantly simplify the calculation of  R\'enyi relative entropies of coherence of  tensor products of quantum states.
Recall that the logarithmic robustness of coherence $C_\lr(\rho)$ quantifies the maximum advantage enabled by a quantum state in the task of  phase discrimination as measured by the logarithm of the ratio of success probabilities
\cite{NapoBCP16,PianCBN16}. The additivity of the logarithmic robustness of coherence thus has an operational implication:  the maximum advantage enabled by a tensor product of quantum states is additive. 
\Thref{thm:AdditivityRREC} also 
implies the additivity of one variant of the geometric coherence $C_\rmG(\rho)$, which coincides with $\underline{C}_{\rmr,1/2}(\rho)$. Incidentally,  the coherence of formation is additive according to \rcite{WintY16}, and the logarithmic $l_1$-norm of coherence is obviously additive. Surprisingly, most useful coherence measures are additive or have additive variants, in sharp contrast with entanglement measures. 

\begin{proof}
	\Eref{eq:AdditivityREC} follows from the formula $C_\rmr(\rho)=S(\rho^{\diag})-S(\rho)$, which is well known. Similarly, \eref{eq:AdditivityRREC} follows from the closed formula of $C_{\rmr,\alpha}(\rho)$ in  \pref{pro:RRECformulaA}. 
	
	To show	\eref{eq:AdditivityRRECb}, let $\rho_{\mrm{MC}}=\cnot[\rho\otimes |0\>\<0|]$. Then 	 
	\begin{align}
	\underline{C}_{\rmr,\alpha}(\rho)=
	\underline{C}_{\rmr,\alpha}(\rho\otimes |0\rangle \langle 0|)
	=\underline{C}_{\rmr,\alpha}(\rho_{\mrm{MC}})
	=-\overline{H}_\alpha^{\uparrow}(A|B)_{\rho_{\mrm{MC}}},
	\end{align}
	where the last equality follows from \thref{thm:EntCohMC}. Now \eref{eq:AdditivityRRECb} is an immediate consequence of \pref{pro:AdditivityRCE}. The same reasoning can also be applied to derive $\eref{eq:AdditivityREC}$ and $\eref{eq:AdditivityRoC}$ as well as \eref{eq:AdditivityRREC} for $\alpha\in [0,2]$.
 In addition, 
	\eref{eq:AdditivityRoC} follows from \eref{eq:AdditivityRRECb} by taking the limit $\alpha \to\infty$. 
\end{proof}

The combination of  \thsref{thm:CohTwoMeasCoincide} and \ref{thm:AdditivityRREC} implies the additivity of the maximum R\'enyi relative entropies of entanglement generated by incoherent operations acting on the system and an incoherent ancilla.
\begin{corollary}\label{L5-7-2}
	\begin{align}
	C_{E_{\rmr}}(\rho_1\otimes \rho_2) &=
	C_{E_{\rmr}}(\rho_1) +C_{E_{\rmr}}(\rho_2), \\
	C_{E_{\rmr,\alpha}}(\rho_1\otimes \rho_2) &=
	C_{E_{\rmr,\alpha}}(\rho_1) +C_{E_{\rmr,\alpha}}(\rho_2)\quad \forall \alpha \in [0,2], \label{6-7-A+}\\
	C_{\underline{E}_{\rmr,\alpha}}(\rho_1\otimes \rho_2)&=
	C_{\underline{E}_{\rmr,\alpha}}(\rho_1)
	+C_{\underline{E}_{\rmr,\alpha}}(\rho_2)\quad \forall \alpha \in \Bigl[\frac{1}{2},\infty\Bigr], \label{6-7-B+}\\
	C_{E_\lr}(\rho_1\otimes \rho_2)&=
	C_{E_\lr}(\rho_1)
	+ C_{E_\lr}(\rho_2).
	\end{align}
\end{corollary}

Further, the combination of  \thref{thm:EntCohMC} and \pref{pro:AdditivityRCE} (or \thref{thm:AdditivityRREC}) implies the additivity of R\'enyi relative entropies of entanglement of maximally correlated states. This result is of intrinsic interest to understanding entanglement properties of maximally correlated states.
\begin{corollary}\label{cor:AdditivityRREE}
	If $\rho_1$ and $\rho_2$ are maximally correlated states, then
	\begin{align}
	{E_{\rmr}}(\rho_1\otimes \rho_2) &=
	{E_{\rmr}}(\rho_1) +{E_{\rmr}}(\rho_2), \\
	{E_{\rmr,\alpha}}(\rho_1\otimes \rho_2) &=
	{E_{\rmr,\alpha}}(\rho_1) +{E_{\rmr,\alpha}}(\rho_2)\quad \forall \alpha \in [0,2], \label{6-7-AC}\\
	{\underline{E}_{\rmr,\alpha}}(\rho_1\otimes \rho_2)&=
	{\underline{E}_{\rmr,\alpha}}(\rho_1)
	+{\underline{E}_{\rmr,\alpha}}(\rho_2)\quad \forall \alpha \in \Bigl[\frac{1}{2},\infty\Bigr], \label{6-7-BC}\\
	{E_\lr}(\rho_1\otimes \rho_2)&=
	{E_\lr}(\rho_1)
	+ {E_\lr}(\rho_2).
	\end{align}
\end{corollary}
This corollary
implies the additivity of the geometric entanglement $E_\rmG(\rho)$, which coincides with $\underline{E}_{\rmr,1/2}(\rho)$, for maximally correlated states. The additivity of an alternative geometric measure was considered in \rcite{ZhuCH10}.
The additivity of the relative entropy of entanglement of maximally correlated states was proven previously in \rcite{Rain99}; the special case of maximally correlated generalized Bell-diagonal states was also considered in \rcite{ZhuCH10}.

\section{\label{sec:bounds}Upper and lower bounds for R\'enyi relative entropies of coherence}
By virtue  of \thref{thm:EntCohMC}, here we derive several nontrivial upper and lower bounds for R\'enyi relative entropies of coherence, including the logarithmic robustness of coherence. Similar bounds apply to R\'enyi relative entropies of entanglement of maximally correlated states.

\begin{theorem}\label{thm:RRECub}
	Any state $\rho$ satisfies
	\begin{align}
	C_{\rmr,\alpha}(\rho)&\leq S_{\frac{1}{\alpha}}\bigl(\rho^{\diag}\bigr)\quad \forall \alpha\in [0,2], \\
	\underline{C}_{\rmr,\alpha}(\rho)&\leq S_{\frac{\alpha}{2\alpha-1}}\bigl(\rho^{\diag}\bigr)\quad \forall \alpha\in \Bigl[ \frac{1}{2},\infty\Bigr], \label{eq:RRECbub}\\
	C_\lr(\rho)&\leq S_{\frac{1}{2}}\bigl(\rho^{\diag}\bigr), \quad 		C_{\rmG}(\rho)\leq -\ln\bigl\|\rho^{\diag}\bigr\|; \label{eq:RoCGCub}
	\end{align}
	all the upper bounds are saturated if $\rho$ is pure.		
\end{theorem}
\begin{proof}
	Let $\rho_{\mrm{MC}}=\cnot[\rho\otimes |0\>\<0|]$. Then
	\begin{equation}
	C_{\rmr,\alpha}(\rho)=C_{\rmr,\alpha}(\rho_{\mrm{MC}})=-H_\alpha^{\uparrow}(A|B)_{\rho_{\mrm{MC}}}\leq S_{\frac{1}{\alpha}}((\rho_{\mrm{MC}})_A)=S_{\frac{1}{\alpha}}\bigl(\rho^{\diag}\bigr)\quad \forall \alpha\in [0,2]
	\end{equation}
	according to \thref{thm:EntCohMC} and \lref{lem:RCElb}. The inequality is saturated if $\rho$ is pure according to \lref{lem:RCElb}.
	By the same token,
	\begin{equation}
	\underline{C}_{\rmr,\alpha}(\rho)=-\overline{H}_\alpha^{\uparrow}(A|B)_{\rho_{\mrm{MC}}}\leq S_{\frac{\alpha}{2\alpha-1}}((\rho_{\mrm{MC}})_A)=S_{\frac{\alpha}{2\alpha-1}}\bigl(\rho^{\diag}\bigr)\quad \forall \alpha\in \Bigl[ \frac{1}{2},\infty\Bigr],
	\end{equation}
	and the inequality is saturated if $\rho$ is pure.
	\Eref{eq:RoCGCub} follows from \eref{eq:RRECbub} by taking the limits $\alpha\to \infty$ and $\alpha\to 1/2$.
\end{proof}

\begin{theorem}\label{thm:RRECabBound}
	Any state $\rho$ satisfies
	\begin{align}
	S_{2-\frac{1}{\alpha}}(\rho\|\rho^{\diag})&\leq 		C_{\rmr,\alpha}(\rho)\leq S_\alpha(\rho\|\rho^{\diag})\quad  \forall  \alpha\in \Bigl[\frac{1}{2},2\Bigr], \label{eq:RRECaBound}
	\\
	\underline{S}_{2-\frac{1}{\alpha}}(\rho\|\rho^{\diag})&\leq 		\underline{C}_{\rmr,\alpha}(\rho)\leq \underline{S}_\alpha(\rho\|\rho^{\diag})\quad \forall \alpha\in \Bigl[ \frac{1}{2},\infty\Bigr], \label{eq:RRECbBound}\\
	\underline{S}_{2}(\rho\|\rho^{\diag})&\leq 		C_\lr(\rho)\leq \underline{S}_\infty(\rho\|\rho^{\diag}); \label{eq:LRoCBound}
	\end{align}
 all the lower bounds in the three equations are saturated if $\rho$ is pure.
\end{theorem}

\begin{proof}
	The upper bounds in \esref{eq:RRECaBound} to \eqref{eq:LRoCBound} are trivial given that $\rho^{\diag}$ is incoherent. To establish the lower bound in \eref{eq:RRECaBound} for $\alpha\in \bigl[\frac{1}{2},2\bigr]$, let $\rho_{\mrm{MC}}=\cnot[\rho\otimes |0\>\<0|]$, then 
	\begin{align}
	&C_{\rmr,\alpha}(\rho)=C_{\rmr,\alpha}(\rho_{\mrm{MC}})=-H_\alpha^{\uparrow}(A|B)_{\rho_{\mrm{MC}}}\geq -H_{2-\frac{1}{\alpha}}^{\downarrow}(A|B)_{\rho_{\mrm{MC}}}=S_{2-\frac{1}{\alpha}}(\rho\|\rho^{\diag}).
	\end{align}
	Here the second and third equalities follow from \thref{thm:EntCohMC} in \sref{sec:Connect} and \lref{lem:RCEMC} below, respectively; the inequality follows from \lref{lem:RCEineq} and is saturated when $\rho$ is pure. The lower bound for $\underline{C}_{\rmr,\alpha}(\rho)$ in \eref{eq:RRECbBound} and the saturation for a pure state can be proved in the same way. \Eref{eq:LRoCBound} follows from \eref{eq:RRECbBound} by taking the limit $\alpha\rightarrow \infty$.
\end{proof}
\Eref{eq:RRECbBound} in \thref{thm:RRECabBound} yields a lower bound for the geometric coherence $C_\rmG(\rho)\geq \underline{S}_0(\rho\|\rho^{\diag})$. 
The bounds for $C_\lr(\rho)$ in \eref{eq:LRoCBound} can be expressed more explicitly as 
\begin{align}\label{eq:LRoClu}
\ln\tr\Bigl\{\bigl[\bigl(\rho^{\diag}\bigr)^{-1/4}\rho\bigl(\rho^{\diag}\bigr)^{-1/4}\bigr]^2\Bigr\}\leq C_\lr(\rho)\leq\ln \bigl\|\bigl(\rho^{\diag}\bigr)^{-1/2}\rho\bigl(\rho^{\diag}\bigr)^{-1/2}\bigr\|.
\end{align}
Here the lower bound improves over the bound $C_\lr(\rho)\geq C_\rmr(\rho)=S(\rho\|\rho^{\diag})$ derived in \cite{RanaPWL16}.
\Eref{eq:LRoClu} implies that
\begin{align}
\tr\Bigl\{\bigl[\bigl(\rho^{\diag}\bigr)^{-1/4}\rho\bigl(\rho^{\diag}\bigr)^{-1/4}\bigr]^2\Bigr\}-1\leq C_{\caR}(\rho)\leq \bigl\|\bigl(\rho^{\diag}\bigr)^{-1/2}\rho\bigl(\rho^{\diag}\bigr)^{-1/2}\bigr\|-1.
\end{align}
In addition, \thsref{thm:RRECub} and
\ref{thm:RRECabBound} enable a simple derivation of R\'enyi relative entropies of coherence of pure states (for certain parameter ranges); cf.~\sref{sec:RREC}. Also, they offer a simple explanation of why the equalities in  \eref{eq:RRECApureAlt} and \eref{eq:RRECBpureAlt} hold.

\begin{lemma}\label{lem:RCEMC}
	Let $\rho_{\mrm{MC}}=\cnot[\rho\otimes |0\>\<0|]$. Then
	\begin{align}
	H_\alpha^{\downarrow}(A|B)_{\rho_{\mrm{MC}}}&=-S_\alpha(\rho\|\rho^{\diag}),\quad 
	\overline{H}_\alpha^{\downarrow}(A|B)_{\rho_{\mrm{MC}}}=-\underline{S}_\alpha(\rho\|\rho^{\diag})\quad \forall \alpha\in [0,\infty].
	\end{align}
\end{lemma}

\begin{proof}According to the definition and \lref{lem:RRErhovsMC} below,
	\begin{align}
	&-H_\alpha^{\downarrow}(A|B)_{\rho_{\mrm{MC}}}=S_\alpha(\rho_{\mrm{MC}}\|I_A\otimes (\rho_{\mrm{MC}})_B)=S_\alpha(\rho_{\mrm{MC}}\|I_A\otimes \rho^{\diag})=	S_\alpha(\rho\|\rho^{\diag}).
	\end{align}
	The other equality in \lref{lem:RCEMC} follows from a similar reasoning.	
\end{proof}

The following lemma is proved in the appendix.
\begin{lemma}\label{lem:RRErhovsMC}
	Let $\rho$ and $\sigma$ be two density matrices on $\caH$ with $\sigma$ being diagonal in the reference basis. Let $\rho_{\mrm{MC}}=\cnot[\rho\otimes |0\>\<0|]$. Then
	\begin{align}
	S_\alpha(\rho_{\mrm{MC}}\|I_A\otimes \sigma) =S_\alpha(\rho\|\sigma),\quad
	\underline{S}_\alpha(\rho_{\mrm{MC}}\|I_A\otimes \sigma)=\underline{S}_\alpha(\rho\|\sigma)\quad \forall \alpha\in [0,\infty].
	\end{align}
\end{lemma}

In view of \thref{thm:EntCohMC},
when $\rho$  is a maximally correlated state, \thsref{thm:RRECub} and \ref{thm:RRECabBound} still hold if R\'enyi relative entropies of coherence are replaced by corresponding R\'enyi relative entropies of entanglement. For example, the following corollary is a consequence of   \thsref{thm:EntCohMC} and  \ref{thm:RRECub}. 
\begin{corollary}
	Any maximally correlated state $\rho$ on $\caH_A\otimes \caH_B$ satisfies
	\begin{align}
	E_{\rmr,\alpha}(\rho)&\leq S_{\frac{1}{\alpha}}\bigl(\rho^{\diag}\bigr)=S_{\frac{1}{\alpha}}(\rho_A)\quad \forall \alpha\in [0,2], \\
	\underline{E}_{\rmr,\alpha}(\rho)&\leq S_{\frac{\alpha}{2\alpha-1}}\bigl(\rho^{\diag}\bigr)=S_{\frac{\alpha}{2\alpha-1}}(\rho_A)\quad \forall \alpha\in \Bigl[ \frac{1}{2},\infty\Bigr], \\
	E_\lr(\rho)&\leq S_{\frac{1}{2}}\bigl(\rho^{\diag}\bigr)=S_{\frac{1}{2}}(\rho_A), \quad 		E_{\rmG}(\rho)\leq -\ln\bigl\|\rho^{\diag}\bigr\|=-\ln\|\rho_A\|.
	\end{align}
	All the upper bounds are saturated if $\rho$ is pure.		
\end{corollary}
Note that this corollary yields a simple derivation of the relative entropy of entanglement and robustness of entanglement of bipartite pure states.

\section{\label{sec:relationR}Relations between R\'enyi relative entropies}
In this section, we determine the condition under which R\'enyi relative entropies are independent of the order parameter $\alpha$. Remember that usually they are monotonically increasing with the order parameter. 
The results presented in this section will be used in  the next section to study the relations between different R\'enyi relative entropies of coherence.

For this purpose, we recall the classical case regarding R\'enyi relative entropies between  two probability distributions
$p$ and $q$ on ${\cal X}$.
We assume that the support of $p$ is included in that of $q$ and define the random variables $\ln p(X)$ and $\ln q(X)$ 
on the support of $p$. Let $\phi(s)$ for $s\geq-1$ be the cumulant generating function of the classical random variable
$\ln p(X)-\ln q(X)$, i.e.,
\begin{equation}
\phi(s):= \ln \rE_{p,X} \exp \bigl\{ s [\ln p(X)-\ln q(X)]\bigr\},
\end{equation}
where
$\rE_{p,X}$ expresses the expectation with respect to the random variable $X$ under the distribution $p$. Then  the R\'enyi relative entropy $S_{1+s}(p\|q)$ can be expressed as $S_{1+s}(p\|q)=\phi(s)/s$. Note that $\phi(0)=0$, we deduce that
\begin{equation}
\phi'(0)=S(p\|q),\quad \phi''(0)=2\lim_{s\to 0}S_{1+s}'(p\|q).
\end{equation}
The first derivative $\phi'(0)$ expresses the expectation of the  variable $\ln p(X)-\ln q(X)$, i.e., 
the relative entropy $S(p\|q)$.
The second derivative $\phi''(0)$ expresses the variance of $\ln p(X)-\ln q(X)$,
which is called the relative varentropy $V(p\|q)$,
\begin{equation}\label{eq:varentropy}
2\lim_{s\to 0}S_{1+s}'(p\|q)=\phi''(0)=V(p\|q):=
\rE_{p,X} [\ln p(X)-\ln q(X)]^2-S(p\|q)^2.
\end{equation}
Incidentally,   $V(p\|q)$ plays an important role in 
the second order analysis and moderate deviation analysis in hypothesis testing  \cite[section~9]{WataH17}\cite{Li14}\cite[(34)]{TomaH13}. In conjunction with the monotonicity of $S_{1+s}(p\|q)$ with $s$, \eref{eq:varentropy} implies the following
proposition.
\begin{proposition}\label{pro:RREequalCon}
The following conditions are equivalent. 
\begin{enumerate}
	\item[(A1)] 	 $S_{1+s}(p\|q)=S(p\|q)$, i.e., $\phi(s)=s\phi'(0)$, for all $s \ge -1$.
	\item[(A2)] 	 $S_{1+s}(p\|q)=S(p\|q)$, i.e., $\phi(s)=s\phi'(0)$, for some $s \ge -1$ with $s\neq0$. 
	\item[(A3)] $\lim_{s\to 0}S_{1+s}'(p\|q)=0$, i.e., $\phi''(0)=0$.
	\item[(A4)]  $p$ is a constant times of $q$ on the support of $p$.
\end{enumerate}
\end{proposition}

Now, we consider the quantum scenario in which $\rho$ and $\sigma$ are two density matrices with  $\supp(\rho)\leq\supp(\sigma)$.   The following analysis
also applies to the case in which $\sigma$ is a positive operator instead of a density matrix. 
Since 
$ {S}_\alpha(\rho\|\sigma)$
and 
$ \underline{S}_\alpha(\rho\|\sigma)$
are  combinations of differentiable functions with respect to $\alpha$,
their derivatives with respect to $\alpha$ are defined and are denoted by
${S}_\alpha'(\rho\|\sigma)$
and  $\underline{S}_\alpha'(\rho\|\sigma)$, respectively. Let $s=\alpha-1$ and
define $\phi(s):= \ln \tr (\rho^{1+s}\sigma^{-s})$ as the analogue of the classical cumulant generating function. Then ${S}_{1+s}(\rho\|\sigma)=\phi(s)/s$ as in the classical case.
Calculation shows that \cite[Exercise 3.5]{Haya17book}
\begin{align}
\phi'(s)
&= 
\frac{\tr \left[\rho^{1+s} ( \ln \rho-\ln \sigma )\sigma^{-s}\right]}{\tr(\rho^{1+s}\sigma^{-s})},
\\
\phi''(s)
&= \frac{\tr \left[\rho^{1+s} ( \ln \rho-\ln \sigma )\sigma^{-s}( \ln \rho-\ln \sigma )\right]}{\tr( \rho^{1+s}\sigma^{-s})}
-\left(\frac{\tr\left[ \rho^{1+s} ( \ln \rho-\ln \sigma )\sigma^{-s}\right]}{\tr (\rho^{1+s}\sigma^{-s})}\right)^2,
\end{align}
which implies that
\begin{align}
\phi'(0)
&= 
\tr [\rho( \ln \rho-\ln \sigma )]
=S(\rho\|\sigma)\\
\phi''(0)
&=V(\rho\|\sigma):=
\tr [\rho ( \ln \rho-\ln \sigma)^2]-S(\rho\|\sigma)^2=
\tr \left\{\rho [ \ln \rho-\ln \sigma -S(\rho\|\sigma) ]^2\right\}.
\end{align}
The relative varentropy $V(\rho\|\sigma)$ in the quantum setting also plays an important role in the second order analysis and moderate deviation analysis in hypothesis testing \cite{ChenH17}\cite{Li14}\cite[(34)]{TomaH13}. As in the classical case, we still have $\phi''(0)=2\lim_{\alpha\to 1}S_\alpha'(\rho\|\sigma)$. 
Suppose $\rho$ and $\sigma$ have spectral decompositions $\rho=\sum_j \lambda_j P_j$ and $\sigma=\sum_k \mu_k Q_k$, where $\lambda_j$ and $\mu_k$ are distinct positive eigenvalues of $\rho$ and $\sigma$, respectively. Then
\begin{align}\label{eq:RREdev}
2\lim_{\alpha\to 1}S_\alpha'(\rho\|\sigma)=\phi''(0)
=\left[
\sum_{jk}a_{jk}\lambda_j\left(\ln\frac{\lambda_j}{\mu_k}\right)^2-\left(\sum_{jk}a_{jk}\lambda_j\ln\frac{\lambda_j}{\mu_k}\right)^2\right],
\end{align}
where $a_{jk}=\tr(P_jQ_k)$, which satisfy $\sum_k a_{jk}=\tr(P_j)$ given that $\supp(\rho)\leq \supp(\sigma)$.

By virtue of \eref{eq:RREdev}, we can prove the following lemma. 
\begin{lemma}\label{lem:RREder}
Suppose $\rho$ is a density matrix and $\sigma$ is a positive operator with $\supp(\sigma)\geq \supp(\rho)$. Then
\begin{equation}\label{eq:RREdevEqual}
	\lim_{\alpha\rightarrow 1} \underline{S}_\alpha'(\rho\|\sigma)=\lim_{\alpha\rightarrow 1} S_\alpha'(\rho\|\sigma).
	\end{equation}	
$\lim_{\alpha\rightarrow 1} \underline{S}_\alpha'(\rho\|\sigma)=0$ and $\lim_{\alpha\rightarrow 1} S_\alpha'(\rho\|\sigma)=0$ iff	$\rho$ commutes with $\sigma$ and is proportional to $\Pi_\rho\sigma$, where $\Pi_\rho$ is the projector onto the support of $\rho$.
\end{lemma} 

\begin{proof}
Since 	$\lim_{\alpha\rightarrow 1} \underline{S}_\alpha(\rho\|\sigma)=\lim_{\alpha\rightarrow 1} S_\alpha(\rho\|\sigma)=S(\rho\|\sigma)$ and $\underline{S}_\alpha(\rho\|\sigma)\leq S_\alpha(\rho\|\sigma)$ for all $\alpha\geq 0$, 
we have
\begin{align}
%\label{eq:RREdevEqual}
\lim_{\alpha\rightarrow 1+0} \underline{S}_\alpha'(\rho\|\sigma)
\le \lim_{\alpha\rightarrow 1+0} S_\alpha'(\rho\|\sigma),\quad
\lim_{\alpha\rightarrow 1-0} \underline{S}_\alpha'(\rho\|\sigma)
 \ge \lim_{\alpha\rightarrow 1-0} S_\alpha'(\rho\|\sigma),
\end{align}
which implies \eref{eq:RREdevEqual}.

Note that the expression  in \eref{eq:RREdev}
may be interpreted as the variance of the variable $\ln\frac{\lambda_j}{\mu_k}$ with respect to the probability distribution composed of the components $a_{jk}\lambda_j$.
If $\rho$ commutes with $\sigma$ and is proportional to $\Pi_\rho\sigma$, then it is straightforward to verify that $\lim_{\alpha\rightarrow 1} \underline{S}_\alpha'(\rho\|\sigma)=\lim_{\alpha\rightarrow 1} S_\alpha'(\rho\|\sigma)=0$; cf.~\eref{eq:RREspecialProof} below. 

Conversely, if $\lim_{\alpha\rightarrow 1} \underline{S}_\alpha'(\rho\|\sigma)=0$ or $\lim_{\alpha\rightarrow 1} S_\alpha'(\rho\|\sigma)=0$, then
 $\ln\frac{\lambda_j}{\mu_k}=c$ for some constant $c$ whenever $a_{jk}\neq0$ (as defined after \eref{eq:RREdev}). In that case, the coefficient matrix $a_{jk}$ has at most one nonzero entry in each row and each column.  On the other hand, by assumption the support of $\rho$ is contained in the support of $\sigma$, which implies that $\sum_k a_{jk}=\sum_k \tr(P_jQ_k)=\tr(P_j)$ for each $j$. 
Therefore, for each spectral projector $P_j$ of $\rho$, there exists a spectral projector $Q_{w(j)}$ of $\sigma$ such that $\tr(P_jQ_{w(j)})=\tr(P_j)$ and $\tr(P_jQ_m)=0$ for all $m\neq w(j)$, where $w$ is an injective map from the spectral projectors of $\rho$ to that of $\sigma$. Consequently, the support of  $P_j$ is contained in the support of $Q_{w(j)}$, so that $\rho$ commutes with $\sigma$. 
Furthermore, $\lambda_j/\mu_{w(j)}$ is a constant according to the above discussion. Therefore, $\rho$ is proportional to $\Pi_\rho\sigma$.
\end{proof}

Now, as the quantum analogue of \pref{pro:RREequalCon}, we derive the following theorem, which is very useful to understanding the relations between R\'enyi relative entropies with different order parameters.
\begin{theorem}\label{thm:RREequalCon}
	Suppose $\rho$ is a density matrix and $\sigma$ is a positive operator with $\supp(\sigma)\geq \supp(\rho)$.
	Then the following conditions are equivalent.
	\begin{enumerate}
		\item[(B1)] 	 $S_\alpha(\rho\|\sigma)=S(\rho\|\sigma)$ for all $\alpha\geq0$.
		
		\item[(B2)] 	 $S_\alpha(\rho\|\sigma)=S(\rho\|\sigma)$ for some $\alpha\geq0$ with $\alpha\neq 1$. 
		
		\item[(B3)] $\lim_{\alpha\rightarrow 1} S_\alpha'(\rho\|\sigma)=0$.

		\item[(B4)] 	 $\underline{S}_\alpha(\rho\|\sigma)=S(\rho\|\sigma)$ for all $\alpha\geq0$.
		
		\item[(B5)] 	 $\underline{S}_\alpha(\rho\|\sigma)=S(\rho\|\sigma)$ for some $\alpha\geq0$ with $\alpha\neq 1$.

		\item[(B6)] $\lim_{\alpha\rightarrow 1} \underline{S}_\alpha'(\rho\|\sigma)=0$.
		
		\item[(B7)]  $\rho$ commutes with $\sigma$ and  is proportional to $\Pi_\rho\sigma$.
	\end{enumerate}
	\end{theorem}

	\begin{proof}
		We shall prove the theorem by establishing the following implications,
		\begin{align*}
		{\rm (B1)}\imply {\rm (B2)} \imply {\rm (B3)} \imply {\rm (B7)} \imply {\rm (B1)},\quad 	
		{\rm (B4)} \imply {\rm (B5)} \imply {\rm (B6)}\imply {\rm (B7)} \imply {\rm (B4)} . 
		\end{align*}			
		Obviously,  (B1) implies (B2). 		
		If $S_\alpha(\rho\|\sigma)=S(\rho\|\sigma)$ for some $\alpha\geq0$ with $\alpha\neq 1$, then $S_\alpha'(\rho\|\sigma)=0$ in the interval $[1,\alpha]$ if $\alpha>1$ or $[\alpha,1]$ if $\alpha<1$, given that $S_\alpha(\rho\|\sigma)$ is monotonically increasing with $\alpha$. Therefore, (B2) implies (B3). The implication ${\rm (B3)} \imply {\rm (B7)}$ is shown in  \lref{lem:RREder}. The implications ${\rm (B4)} \imply {\rm (B5)} \imply {\rm (B6)}\imply {\rm (B7)}$ follow from a similar reasoning.

For the  implications ${\rm (B7)} \imply {\rm (B1)}$ and ${\rm (B7)}\imply {\rm (B4)}$, note that $\underline{S}_\alpha(\rho\|\sigma)=S_\alpha(\rho\|\sigma)$ because $\rho$ commutes with $\sigma$.
		Meanwhile, the condition (B7) implies that $\Pi_\rho\sigma=c\rho$ for some constant $c>0$, so that
		\begin{align}
	\tr(\rho^\alpha \sigma^{1-\alpha})=	\tr(\rho^\alpha \Pi_\rho \sigma^{1-\alpha})
		=\tr[\rho^\alpha (\Pi_\rho \sigma)^{1-\alpha}]=c^{1-\alpha}\tr(\rho)=c^{1-\alpha}
		\end{align}
		for $\alpha\geq0$.
		Therefore, 
		\begin{equation}\label{eq:RREspecialProof}
		\underline{S}_\alpha(\rho\|\sigma)=S_\alpha(\rho\|\sigma)=-\ln c \quad \forall \alpha\in [0,\infty],
		\end{equation}
		which implies (B1) and (B4). 
	\end{proof}

As  applications  of \eref{eq:RREdev} and \thref{thm:RREequalCon},
here we reproduce several well-known folklore results concerning R\'enyi  entropies based on the observation $S_\alpha(\rho)=-S_\alpha(\rho\| I)$. 
Setting $\sigma=I$ in \eref{eq:RREdev} yields
\begin{equation}\label{eq:REdev}
	\lim_{\alpha\rightarrow 1} S_\alpha'(\rho)=-\frac{1}{2}\left[
	\sum_{j}m_j\lambda_j\left(\ln\lambda_j\right)^2-\left(\sum_{j}m_j\lambda_j\ln\lambda_j\right)^2\right],
	\end{equation}
where $\lambda_j$ are the distinct eigenvalues of $\rho$ and $m_j$ are the corresponding multiplicities.	 
It follows  that	$\lim_{\alpha\rightarrow 1} S_\alpha'(\rho)=0$ iff	all nonzero eigenvalues of $\rho$ are equal, that is, $\rho$  is proportional to a projector.

\Thref{thm:RREequalCon} has an analogue for R\'enyi  entropies. 
\begin{corollary}
	The following statements concerning a density matrix $\rho$ are equivalent.
	\begin{enumerate}
		\item[(C1)] 	 $S_\alpha(\rho)=S(\rho)$ for all $\alpha\geq0$.
		
		\item[(C2)] 	 $S_\alpha(\rho)=S(\rho)$ for some $\alpha\geq0$ with $\alpha\neq 1$. 
		
		\item[(C3)] $\lim_{\alpha\rightarrow 1} S_\alpha'(\rho)=0$.

		\item[(C4)]  $\rho$ is proportional to a projector.
	\end{enumerate}
\end{corollary}

\section{\label{sec:relationRREC}Relations between R\'enyi relative entropies of coherence}

By virtue of the results presented in  previous sections, here we clarify order relations between
different R\'enyi relative entropies of coherence, including the logarithmic robustness of coherence.
We then determine all states whose relative entropy of coherence (or distillable coherence) is equal to the logarithmic robustness of coherence or
geometric coherence. These results will be useful in understanding the relation between exact coherence distillation and asymptotic coherence distillation as discussed in \sref{sec:ExactDistill}.

First, the inequality \eref{eq:RREabOrder} implies that
\begin{align}
C_{\rmr,\alpha}(\rho) &\ge \underline{C}_{\rmr,\alpha}(\rho) \quad \forall \alpha\in [0,\infty]. \label{eq:RRECabOrder}
\end{align}
The following theorem establishes inequalities in the opposite direction.
\begin{theorem}\label{thm:RRECabOrder}
	Any state $\rho$ satisfies
	\begin{align}
	\underline{C}_{\rmr,\alpha}(\rho) &\geq C_{\rmr, 2-\frac{1}{\alpha}}(\rho) \quad \forall \alpha\in \Bigl[ \frac{1}{2},\infty\Bigr], \label{eq:RRECboundB}\\
	C_\lr(\rho)&\geq C_{\rmr, 2}(\rho) =\ln \Biggl[\sum_j \biggl(\sum_k |\rho_{jk}|^2\biggr)^{1/2}\Biggr]^2, \label{eq:RoCbound}
	\end{align}
	and the two inequalities are saturated when $\rho$ is pure.	
\end{theorem}
\begin{proof}
	Let $\rho_{\mrm{MC}}=\cnot[\rho\otimes |0\>\<0|]$. Then 
	\begin{align}
	&\underline{C}_{\rmr,\alpha}(\rho)=\underline{C}_{\rmr,\alpha}(\rho_{\mrm{MC}})=-\overline{H}_\alpha^{\uparrow}(A|B)_{\rho_{\mrm{MC}}} \geq
	-H_{2-\frac{1}{\alpha}}^{\uparrow}(A|B)_{\rho_{\mrm{MC}}}=C_{\rmr, 2-\frac{1}{\alpha}}(\rho_{\mrm{MC}})=C_{\rmr, 2-\frac{1}{\alpha}}(\rho)
	\end{align}
	according to \thref{thm:EntCohMC} and \lref{lem:RCEineq}. In addition, \lref{lem:RCEineq} implies that the inequality is saturated when $\rho$ is pure, which can also be verified explicitly by virtue of   \pref{pro:RRECpure}.
	Taking the limit $\alpha\to \infty $ in \eref{eq:RRECboundB} and applying \eref{eq:RRECr2}, we obtain \eref{eq:RoCbound}. Again, 
	the inequality $C_\lr(\rho)\geq C_{\rmr, 2}(\rho)$ is saturated when $\rho$ is pure. 
\end{proof}
\Eref{eq:RRECboundB} in \thref{thm:RRECabOrder} yields a lower bound for the geometric coherence, 
\begin{equation}
C_\rmG(\rho)\geq C_{\rmr, 0}(\rho)=-\ln \bigl\|(\Pi_\rho)^{\diag}\bigr\|,
\end{equation}
where the formula for $C_{\rmr, 0}(\rho)$ comes from  \eref{eq:RRECformulaA0} and $\Pi_\rho$ is the projector onto the support of $\rho$. This in turn leads to a lower bound for the other variant of the geometric coherence $\tilde{C}_\rmG(\rho)$, that is,
\begin{equation}
\tilde{C}_\rmG(\rho)\geq 1- \bigl\|(\Pi_\rho)^{\diag}\bigr\|.
\end{equation} 
\Eref{eq:RoCbound} improves over the bound $C_\lr(\rho)\geq C_{\rmr}(\rho)$ known previously \cite{RanaPWL16}, note that $C_{\rmr, 2}(\rho)\geq C_\rmr(\rho)$. As a corollary, we get a lower bound for the robustness of coherence,
\begin{equation}
C_{\caR}(\rho)\geq\Biggl[\sum_j \biggl(\sum_k |\rho_{jk}|^2\biggr)^{1/2}\Biggr]^2-1.
\end{equation}

By virtue of \thref{thm:RRECabOrder} and the inequality $C_\lr(\rho)\leq C_\rmL(\rho)$ \cite{PianCBN16}, we can derive a universal upper bound for all R\'enyi relative entropies of coherence.
\begin{corollary}\label{cor:RRECuub}
Any state $\rho$ satisfies
\begin{align}
C_{\rmr,\alpha}(\rho)&\leq C_\lr(\rho)\leq C_\rmL(\rho)\quad \forall \alpha\in [0,2],\\
\underline{C}_{\rmr,\alpha}(\rho)&\leq C_\lr(\rho)\leq C_\rmL(\rho)\quad \forall \alpha \in[0,\infty].
\end{align}
\end{corollary}
In conjunction with \eref{eq:RRECformulaA}, \crref{cor:RRECuub} leads to an interesting inequality,
\begin{equation}
\frac{1}{\alpha-1}\ln \bigl\|(\rho^\alpha)^{\diag}\bigr\|_{1/\alpha}\leq C_\rmL(\rho)\quad \forall \alpha\in [0,2],
\end{equation}
which is applicable to any density matrix. It is
equivalent to the following inequality,
\begin{equation} \Bigl(\bigl\|(\rho^\alpha)^{\diag}\bigr\|_{1/\alpha}\Bigr)^{\frac{1}{\alpha-1}}\leq C_{l_1}(\rho)+1=\sum_{j,k} |\rho_{jk}|\quad  \forall \alpha\in [0,2].
\end{equation}

If $\rho$ is a maximally correlated state, then the logarithmic negativity is equal to the logarithmic $l_1$-norm of coherence, that is,
$\caN_\rmL(\rho)=C_\rmL(\rho)$  \cite{RanaPWL16,ZhuHC17}. By virtue of  \thref{thm:EntCohMC} and \crref{cor:RRECuub}, we can derive a universal upper bound for all R\'enyi relative entropies of entanglement of maximally correlated states.
\begin{corollary}
	Any maximally correlated state $\rho$ satisfies
	\begin{align}
	E_{\rmr,\alpha}(\rho)&\leq E_\lr(\rho)\leq \caN_\rmL(\rho)\quad \forall \alpha\in [0,2],\\
	\underline{E}_{\rmr,\alpha}(\rho)&\leq E_\lr(\rho)\leq \caN_\rmL(\rho)\quad \forall \alpha \in[0,\infty].
	\end{align}
\end{corollary}
Note that the two equations above still hold if $\rho$ is subjected to any local unitary transformation.

Now, using \thsref{thm:RRECabBound}, \ref{thm:RREequalCon}, and \ref{thm:RRECabOrder},
we determine all states whose
 relative entropy of coherence (or distillable coherence \cite{WintY16})  coincides with the logarithmic robustness of coherence or geometric coherence.
 
\begin{theorem}\label{thm:EntropyRobustEqual}
	The following conditions are equivalent.
	\begin{itemize}
		\item[(D1)]
		$C_\lr(\rho)=\underline{C}_{\rmr,\infty}(\rho)=C_{\rmr}(\rho)$. 
		\item[(D2)]
		$\underline{C}_{\rmr,\alpha}(\rho)=C_\rmr(\rho)$
		for some $\alpha\geq 1/2$ with $\alpha\neq 1$.
		
		\item[(D3)]
		$\underline{C}_{\rmr,\alpha}(\rho)=C_\rmr(\rho)$ for all $\alpha\geq1/2$.
		
		\item[(D4)]
		$C_{\rmr,\alpha}(\rho)=C_\rmr(\rho)$ for all $\alpha\geq0$.
		
		\item[(D5)]
		$C_{\rmr,\alpha}(\rho)=C_\rmr(\rho)$
		for some $\alpha\geq0$ with $\alpha\neq 1$.
		
		\item[(D6)]	$\rho$ commutes with $\rho^{\diag}$ and is proportional to $\Pi_\rho \rho^{\diag}$. 
 	\end{itemize}
\end{theorem}
Different R\'enyi relative entropies of coherence are interesting in different contexts and have different operational interpretations. For example, the relative entropy of coherence is equal to the distillable coherence \cite{WintY16}, while the geometric coherence upper bounds the exact coherence distillation rate (see \sref{sec:ExactDistill}).
Therefore, \thref{thm:EntropyRobustEqual} is instructive to understanding  the connections between  different operational tasks in which R\'enyi relative entropies of coherence play certain roles. For example, \thref{thm:EntropyRobustEqual} is helpful to clarifying the relation between exact coherence distillation and asymptotic coherence distillation.

The combination of \thref{thm:EntropyRobustEqual} and \crref{cor:RRECuub} yields the following result.
\begin{corollary}\label{cor:RECLL}
	If $\rho$ saturates the inequality $C_{\rmr}(\rho)\leq C_\rmL(\rho)$, then
$\rho$ commutes with $\rho^{\diag}$ and is proportional to $\Pi_\rho \rho^{\diag}$.
\end{corollary}
As an implication of \thref{thm:EntropyRobustEqual} and \crref{cor:RECLL}, when $\rho$ is  pure,  $C_\lr(\rho)=C_{\rmr}(\rho)$ iff all nonzero elements of $\diag(\rho)$ are equal, in which case $\rho$ is either incoherent or maximally coherent on the support of $\rho^{\diag}$ (here $\diag(\rho)$ is a vector, while  $\rho^{\diag}$ is a diagonal matrix). Similarly, when $\rho$ is a qubit state,  $C_\lr(\rho)=C_{\rmr}(\rho)$ iff $\rho$ is incoherent or maximally coherent. The same is true if $C_\lr(\rho)$ is replaced by $C_\rmL(\rho)$ given that  $C_\rmL(\rho)=C_\lr(\rho)$ in both cases. In general, incoherent states and maximally coherent states can satisfy the conditions in \thref{thm:EntropyRobustEqual}, but they are not the only candidates. For example, the conditions can also be satisfied by a weighted direct sum of two  maximally coherent states, say
 \begin{equation}
 \rho=p_1(|\psi\>\<\psi|)+p_2 (|\varphi\>\<\varphi|),
 \end{equation}
 where 
 \begin{equation}
  |\psi\>=\frac{1}{\sqrt{2}}(|0\>+|1\>),\quad  |\varphi\>=\frac{1}{\sqrt{2}}(|2\>+|3\>),\quad  0\leq p_1, p_2\leq 1,\quad  p_1+p_2=1.
 \end{equation}

When $\rho$ is a maximally correlated state, \thref{thm:EntropyRobustEqual} and \crref{cor:RECLL} still hold if R\'enyi relative entropies of coherence are replaced by the corresponding R\'enyi relative entropies of entanglement, while the logarithmic $l_1$-norm of coherence is replaced by the logarithmic negativity. For example, the following corollary is the analogue of \crref{cor:RECLL}.
\begin{corollary}
	If $\rho$ is a maximally correlated state that saturates the inequality $E_{\rmr}(\rho)\leq \caN_\rmL(\rho)$, then
	$\rho$ commutes with $\rho^{\diag}$ and is proportional to $\Pi_\rho \rho^{\diag}$.
\end{corollary}

\begin{proof}[Proof of \thref{thm:EntropyRobustEqual}]
	We shall prove \thref{thm:EntropyRobustEqual} by establishing the following implications,
	\begin{align*}
	&{\rm (D4)}\imply {\rm (D3)} \imply {\rm (D1)} \imply {\rm (D2)} \imply {\rm (D5)} \imply {\rm (D6)} \imply {\rm (D4)}. 
	\end{align*}
	The implications ${\rm (D4)}\imply {\rm (D3)}$ and ${\rm (D2)} \imply {\rm (D5)}$ follow from \thref{thm:RRECabOrder}, \eref{eq:RRECabOrder}, and the monotonicity of $C_{\rmr,\alpha}, \underline{C}_{\rmr,\alpha}$ with $\alpha$. The implications ${\rm (D3)} \imply {\rm (D1)}$ and ${\rm (D1)} \imply {\rm (D2)}$ are trivial. The implication ${\rm (D6)} \imply {\rm (D4)}$ follows from \lref{lem:RRECspecial} below.
	
	It remains to show the implication ${\rm (D5)} \imply {\rm (D6)}$. If $C_{\rmr,\alpha}(\rho)=C_\rmr(\rho)$ for some $\alpha<1$, then $S_\alpha(\rho\|\rho^{\diag})=C_\rmr(\rho)=S(\rho\|\rho^{\diag})$. If $C_{\rmr,\alpha}(\rho)=C_\rmr(\rho)$ for some $\alpha>1$, then 
	$S_{2-\frac{1}{\alpha}}(\rho\|\rho^{\diag})=C_\rmr(\rho)=S(\rho\|\rho^{\diag})$ according to \thref{thm:RRECabBound}. Therefore, $\rho$ commutes with $\rho^{\diag}$ and is proportional to $\Pi_\rho \rho^{\diag}$ according to \thref{thm:RREequalCon}.
\end{proof}

In the rest of this section, we prove a lemma used in the proof of \thref{thm:EntropyRobustEqual}.
\begin{lemma}\label{lem:RRECspecial}
Suppose $\rho$ is  a density matrix that  commutes with
	 $\rho^{\diag}$ and satisfies  $\Pi_\rho \rho^{\diag}=c\rho$ for some positive constant $c$. Then
	\begin{align}
	C_\lr(\rho)=	C_{\rmr,\alpha}(\rho)=	\underline{S}_\alpha(\rho\|\rho^{\diag})=S_\alpha(\rho\|\rho^{\diag})&=-\ln c\quad \forall \alpha\in [0,\infty], \label{eq:RREspecial}\\
	\underline{C}_{\rmr,\alpha}(\rho)&=-\ln c \quad \forall \alpha\in \Bigl[ \frac{1}{2},\infty\Bigr].
	\end{align} 
\end{lemma}

\begin{proof}
The equalities $\underline{S}_\alpha(\rho\|\rho^{\diag})=S_\alpha(\rho\|\rho^{\diag})=-\ln c$ follow from \eref{eq:RREspecialProof}. To prove other equalities in the lemma,
let $\rho=\sum_j \lambda_j P_j$ be the spectral decomposition of $\rho$ with $\lambda_j>0$.  
If $\rho$   commutes with $\rho^{\diag}$ and satisfies  $\Pi_\rho \rho^{\diag}=c\rho$, then  $P_j^{\diag}$ have mutually orthogonal supports and all nonzero entries of $P_j^{\diag}$ are equal to $c$.
Suppose $P_j^{\diag}$ have $n_j$ nonzero entries, then $\sum_j n_j \lambda_j c=\tr\rho=1$, so that $\sum_j n_j \lambda_j=1/c$. 
According to  \pref{pro:RRECformulaA},
\begin{align}
	C_{\rmr,\alpha}(\rho)&=\frac{1}{\alpha-1}\ln \bigl\|(\rho^\alpha)^{\diag}\bigr\|_{1/\alpha}=\frac{1}{\alpha-1}\ln \Biggl\|\biggl(\sum_j\lambda_j^\alpha P_j\biggr)^{\diag}\Biggr\|_{1/\alpha}\nonumber \\
	&=\frac{1}{\alpha-1}\ln \left(\sum_j n_j \lambda_jc^{1/\alpha} \right)^\alpha
	 =\frac{1}{\alpha-1}\ln c^{1-\alpha}=-\ln c \quad \forall \alpha\in [0,\infty),
	\end{align}
which further implies that 	$C_{\rmr,\infty}(\rho)=-\ln c$. 
According to \thref{thm:RRECabBound}, 
	\begin{equation}
	-\ln c=\underline{S}_{\rmr,2-\frac{1}{\alpha}}(\rho\|\rho^{\diag})\leq \underline{C}_{\rmr,\alpha}(\rho)\leq \underline{S}_{\rmr,\alpha}(\rho\|\rho^{\diag})=-\ln c\quad \forall \alpha\in\Bigl[\frac{1}{2},\infty\Bigr),
	\end{equation}
	which implies that $\underline{C}_{\rmr,\alpha}(\rho)=-\ln c$ for $\alpha\geq 1/2$. Alternatively, this result can  be derived from \thref{thm:RRECabOrder} and the equality $C_{\rmr,\alpha}(\rho)=-\ln c$.	
	Taking the limit $\alpha\to \infty$ yields the equality $C_\lr(\rho)=-\ln c$.
\end{proof}

\section{\label{sec:ExactDistill}Exact coherence distillation}

It is known that R\'{e}nyi relative entropies of entanglement upper bound the exact distillation rate of entanglement \cite[lemma 8.15]{Haya17book}.
Moreover, the bounds based on $E_{\rmr,0}$ and $\underline{E}_{\rmr,1/2}$ are saturated in the case of pure states \cite{HayaKMM03,Haya06}\cite[Exercise 8.32]{Haya17book}.
In this section we show that R\'enyi relative entropies of coherence play the same role in exact coherence distillation
as R\'{e}nyi relative entropies of entanglement play in exact entanglement distillation.

Exact coherent distillation is a procedure for producing perfect maximally coherent states from  partially coherent states as illustrated in \fref{fig:ECD}. In other words, the goal is to generate  maximally coherent states with exactly zero error. By contrast, in conventional asymptotic coherence distillation, the goal is to generate  maximally coherent states with a small error that goes to zero asymptotically.
 For a given state $\rho$, we define the exact coherence distillation length 
$L_{\rme,\rmc}(\rho)$ as
\begin{align}\label{eq:CDL}
L_{\rme,\rmc}(\rho):=
\max\{ L |
\exists \Lambda_\rmi, \;
\Lambda_\rmi(\rho)=|\Phi_{\rmc,L}\rangle \langle \Phi_{\rmc,L}|
\},
\end{align}
where $|\Phi_{\rmc,L}\rangle
:=\sum_{j=0}^{L-1}\frac{1}{\sqrt{L}}|j\rangle$ is a maximally coherent state in dimension $L$ \cite{BaumCP14},
and $\Lambda_\rmi$ is an incoherent operation.
Then, we define the asymptotic exact coherent distillation rate
$R_{\rme,\rmc}(\rho)$ as
\begin{align}
R_{\rme,\rmc}(\rho):= \lim_{n \to \infty}\frac{1}{n}\ln L_{\rme,\rmc}(\rho^{\otimes n}).
\end{align}

\begin{figure}
	\centering
	\includegraphics[width=12cm]{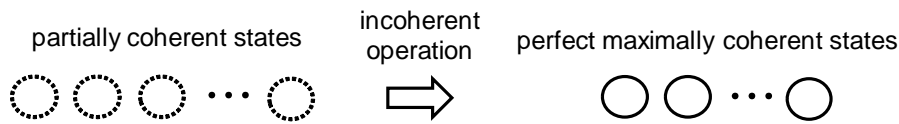}
	\caption{\label{fig:ECD}Exact coherence distillation.
This procedure  produces perfect maximally coherent states from partially coherent states.}
\end{figure}
\begin{lemma}\label{pro:ExactDistill}
	\begin{align}
	\ln L_{\rme,\rmc}(\rho)\leq R_{\rme,\rmc}(\rho)&\leq C_{\rmr,\alpha}(\rho) \quad \forall \alpha\in [0,2],\\
	\ln L_{\rme,\rmc}(\rho)\leq R_{\rme,\rmc}(\rho)&\leq \underline{C}_{\rmr,\alpha}(\rho) \quad \forall \alpha\in \Bigl[ \frac{1}{2},\infty\Bigr].
	\end{align}
\end{lemma}
\begin{proof}
	According to the definition of $L_{\rme,\rmc}(\rho)$, it is straightforward to verify that $L_{\rme,\rmc}(\rho^{\otimes n})\geq L_{\rme,\rmc}(\rho)^n$. Therefore,
	\begin{align}
	R_{\rme,\rmc}(\rho)&= \lim_{n \to \infty}\frac{1}{n}\ln L_{\rme,\rmc}(\rho^{\otimes n})\geq \lim_{n \to \infty}\frac{1}{n}\ln L_{\rme,\rmc}(\rho)^n=\ln L_{\rme,\rmc}(\rho).
	\end{align}
	
	Let $C$ be any coherent measure that does not increase under incoherent operations. Then $C(|\Phi_{\rmc,L}\rangle \langle \Phi_{\rmc,L}|)\leq C(\rho)$ whenever $|\Phi_{\rmc,L}\rangle \langle \Phi_{\rmc,L}|$ can be generated from $\rho$ by incoherent operations. If, in addition, $C$ satisfies the normalization condition $C(|\Phi_{\rmc,L}\rangle \langle \Phi_{\rmc,L}|)=\ln L$, which is the case for all the coherent measures that appear in \lref{pro:ExactDistill}, then $\ln L_{\rme,\rmc}(\rho)\leq C(\rho)$. Therefore,
	\begin{align}
	R_{\rme,\rmc}(\rho)\leq \lim_{n \to \infty}\frac{1}{n}C(\rho^{\otimes n}). 
	\end{align}
	Now \lref{pro:ExactDistill} follows from the fact that $C_{\rmr,\alpha}$ for $ \alpha\in [0,2]$ and $\underline{C}_{\rmr,\alpha}(\rho)$ for $\alpha\in \bigl[ \frac{1}{2},\infty\bigr]$ are additive according to \thref{thm:AdditivityRREC}.
\end{proof}
Recall that both $C_{\rmr,\alpha}$ and $\underline{C}_{\rmr,\alpha}(\rho)$ are monotonically increasing with $\alpha$ and that $C_{\rmr,0}(\rho) \le \underline{C}_{\rmr,1/2}(\rho)$ according to \eref{eq:RRECboundB}. So the bound $C_{\rmr,0}(\rho)$ on the exact distillation rate is the best among all bounds 
based on R\'enyi relative entropies of coherence.
Actually, this bound is saturated when $\rho$ is pure, in which case $C_{\rmr,0}(\rho)= \underline{C}_{\rmr,1/2}(\rho)$.

\begin{theorem}\label{thm:ExactDistillpure}
	Suppose $\rho=|\psi\>\<\psi|$ is a pure state. Then $L_{\rme,\rmc}(\rho)=\lfloor 1/p_{\max}\rfloor$ and $R_{\rme,\rmc}(\rho)=-\ln(p_{\max})=C_{\rmr,0}(\rho)$, where $p_{\max}=p_{\max}(\rho):=\|\diag(\rho)\|_\infty$. 
\end{theorem}
\begin{proof}
	A pure state $\rho$ can be transformed to another pure state $\sigma$ under incoherent operations iff $\diag(\rho)$ is majorized by $\diag(\sigma)$ \cite{DuBG15,WintY16,ChitG16a,ZhuMCF17} (the same is true if we consider strictly incoherent operations). In addition, $\diag(\rho)$ is majorized by $\diag(|\Phi_{\rmc,L}\rangle \langle \Phi_{\rmc,L}|)$ iff $p_{\max}\leq 1/L$. Therefore, $L_{\rme,\rmc}(\rho)=\lfloor 1/p_{\max}\rfloor$, 
	\begin{align}
	R_{\rme,\rmc}(\rho)&=\frac{1}{n}\lim_{n \to \infty}\ln\left \lfloor (p_{\max}(\rho^{\otimes n}))^{-1}\right\rfloor=\frac{1}{n}\lim_{n \to \infty}\ln\left \lfloor (p_{\max})^{-n}\right\rfloor=-\ln(p_{\max})=C_{\rmr,0}(\rho).
	\end{align}
\end{proof}

Note that \lref{pro:ExactDistill} and \thref{thm:ExactDistillpure} still hold if the operation $\Lambda_\rmi$ in the definition of $L_{\rme,\rmc}(\rho)$ in \eref{eq:CDL} is only required to be incoherence-preserving instead of being incoherent. In this case, the current proof of \lref{pro:ExactDistill} still applies after replacing incoherent operations with incoherence-preserving operations. The current proof of \thref{thm:ExactDistillpure} implies that $L_{\rme,\rmc}(\rho)\geq\lfloor 1/p_{\max}\rfloor$ and 
$R_{\rme,\rmc}(\rho)\geq-\ln(p_{\max})=C_{\rmr,0}(\rho)$, while the opposite inequalities follow from \lref{pro:ExactDistill}. Therefore, for pure states, the exact coherence  distillation rate (length)
remains the same under three distinct classes of operations, namely, strictly incoherent operations, incoherent operations, and incoherence-preserving operations.

In general, the exact distillation rate $R_{\rme,\rmc}(\rho)$ is smaller than the distillable coherence, which is equal to the relative entropy of coherence $C_{\rmr}(\rho)$ \cite{WintY16}. Therefore,  exact distillation requires more resources than distillation with negligible small error even asymptotically under incoherence-preserving operations.
Consequently, the exact distillation rate of coherence is in general smaller than the coherence cost.

A necessary condition for saturating the inequality $R_{\rme,\rmc}(\rho)\leq C_{\rmr}(\rho)$ can be derived from \thref{thm:EntropyRobustEqual} and \lref{pro:ExactDistill}.
\begin{corollary}
	If the exact distillation rate of coherence is equal to the distillable coherence, that is, if
	the bound $R_{\rme,\rmc}(\rho)\leq C_{\rmr}(\rho)$ is saturated, then $\rho$  commutes with $\rho^{\diag}$ and is proportional to $\Pi_\rho \rho^{\diag}$. 
\end{corollary}
According to this corollary or \thref{thm:ExactDistillpure}, when $\rho$ is a pure state, the inequality $R_{\rme,\rmc}(\rho)\leq C_{\rmr}(\rho)$ is saturated iff $\rho$ is incoherent or  maximally coherent on the support of $\rho^{\diag}$. Similarly, when $\rho$ is a qubit state, the inequality $R_{\rme,\rmc}(\rho)\leq C_{\rmr}(\rho)$
is saturated iff $\rho$ is incoherent or maximally coherent.

\section{\label{sec:Sum}Summary}

We proved that R\'enyi relative entropies of coherence and R\'enyi relative entropies of entanglement are both equal to the corresponding R\'enyi conditional entropies for maximally correlated states. By virtue of this observation and  the generalized CNOT gate, we established an operational one-to-one mapping
between entanglement measures and coherence measures based on  R\'enyi relative entropies.  In particular, every R\'enyi relative entropy of coherence is equal to the maximum R\'enyi relative entropy of entanglement generated by incoherence-preserving operations (or incoherent operations) acting on the system and an incoherent ancilla. 
These results significantly strengthen the connection between the resource theory of coherence and that of entanglement. They are also useful to understanding the properties of maximally correlated states themselves.
We then proved that all R\'enyi relative entropies of coherence, including the logarithmic robustness of coherence, are additive. Accordingly, 
all R\'enyi relative entropies of entanglement  are additive for maximally correlated states. In addition, we derived several nontrivial bounds on R\'enyi relative entropies of coherence and logarithmic robustness of coherence, which improve over bounds known in the literature, including the inequality $C_\rmr(\rho)\leq C_\lr(\rho)$ between the relative entropy of coherence and  logarithmic robustness of coherence. Furthermore, we determined all states whose  relative entropy of coherence (or distillable coherence) is equal to the logarithmic robustness of coherence or geometric coherence. As an application, we provided an upper bound for the exact coherence distillation rate based on a special R\'enyi relative entropy of coherence, which is saturated for pure states.

%%%%%%%%%%%%%%%%%%%%%%%%%%%%%%%%%%%%%%%%%%%%%%%%%%%%%%%%%%%%%%%%%%%%%%%%%%%%%%%%%%%%%%%%%%%%%%%%%%%%%%%%%%%%%%%%%%%%%%%%%%%%%%%%%%%%%%%%%%%%%%%%%%%%%%%%%%%%%%%%%%%%%%%%%%%%%%%%%%%%%%%%%%%%%%%%%%%%%%%%%%%%%%%%%%%%%%%%%%%%%%%%%%%%%%%%%

%\acknowledgments
\ack
We are grateful to a referee for careful  comments and constructive suggestions.
HZ acknowledges financial support
from the Excellence
Initiative of the German Federal and State Governments
(ZUK~81) and the DFG.
MH was supported in part by JSPS Grants-in-Aid for Scientific Research (A) No. 17H01280 and (B) No. 16KT0017 and Kayamori Foundation of Informational Science Advancement.
LC was supported by Beijing Natural Science Foundation (4173076), the National Natural Science Foundation (NNSF) of China (Grant No. 11501024), and the Fundamental Research Funds for the Central Universities (Grants No. KG12001101, No. ZG216S1760, and No. ZG226S17J6).

\appendix

\section{Alternative proof of \pref{pro:RRECformulaA}}
In this appendix we give an alternative proof of  \pref{pro:RRECformulaA} for $\alpha\in [0,2]$
by virtue of \thref{thm:EntCohMC} and \eref{eq:CEaformula}. 
\begin{proof}
	Let $\rho_{\mrm{MC}}=\cnot[\rho\otimes |0\>\<0|]$, then 
	\begin{align}
	C_{\rmr,\alpha}(\rho)&=C_{\rmr,\alpha}(\rho_{\mrm{MC}})=-H_\alpha^{\uparrow}(A|B)_{\rho_{\mrm{MC}}}=\frac{\alpha}{\alpha-1}\ln\tr\bigl\{[\tr_A(\rho_{\mrm{MC}}^\alpha)]^{1/\alpha}\bigr\}\nonumber\\
	&=\frac{1}{\alpha-1}\ln \bigl\|(\rho^\alpha)^{\diag}\bigr\|_{1/\alpha}.
	\end{align}
	Here the second inequality follows from \thref{thm:EntCohMC}, the third one from \eref{eq:CEaformula}, and the last one from the observation that
	\begin{equation}
	\tr_A(\rho_{\mrm{MC}}^\alpha)=\sum_j (\rho^{\alpha})_{jj} |j\>\<j|=(\rho^\alpha)^{\diag}.
	\end{equation}
\end{proof}

\section{Proof of \lref{lem:RRErhovsMC}}
\begin{proof}
	According to the definition of $S_\alpha$ in \eref{eq:RRE},
	\begin{align}
	S_\alpha\left(\rho_{\mrm{MC}}\|I_A\otimes \sigma\right)&=\frac{1}{\alpha-1}\ln\tr \left\{\rho_{\mrm{MC}}^\alpha \left[I_A\otimes \sigma^{1-\alpha}\right]\right\}=\frac{1}{\alpha-1}\ln\tr \left\{(\rho^\alpha)^{\diag} \sigma^{1-\alpha}\right\}\nonumber \\
	&=\frac{1}{\alpha-1}\ln\tr (\rho^\alpha \sigma^{1-\alpha})=S_\alpha(\rho\|\sigma),
	\end{align}
	where the third equality follows from the assumption that $\sigma$ is diagonal in the reference basis. Similarly,
	\begin{align}
	\underline{S}_\alpha\left(\rho_{\mrm{MC}}\|I_A\otimes \sigma\right)&=\frac{1}{\alpha-1}\ln\tr \left\{\bigl[ \bigl(I_A\otimes \sigma^{\frac{1-\alpha}{2\alpha}}\bigr)\rho_{\mrm{MC}} \bigl(I_A\otimes \sigma^{\frac{1-\alpha}{2\alpha}}\bigr)\bigr]^\alpha\right\}\nonumber \\
	&=\frac{1}{\alpha-1}\ln\tr \left\{\left[\sum_{jk} \sigma_{jj}^{\frac{1-\alpha}{2\alpha}}\rho_{jk}\sigma_{kk}^{\frac{1-\alpha}{2\alpha}} (|jj\>\<kk|) \right]^\alpha\right\}\nonumber \\
	&=\frac{1}{\alpha-1}\ln\tr \left\{\left[\sum_{jk} \sigma_{jj}^{\frac{1-\alpha}{2\alpha}}\rho_{jk}\sigma_{kk}^{\frac{1-\alpha}{2\alpha}} (|j\>\<k|) \right]^\alpha\right\}\nonumber \\
	&=\frac{1}{\alpha-1}\ln\tr \bigl[\bigl( \sigma^{\frac{1-\alpha}{2\alpha}}\rho\sigma^{\frac{1-\alpha}{2\alpha}} \bigr)^\alpha\bigr]=\underline{S}_\alpha(\rho\|\sigma).
	\end{align}
\end{proof}

\section*{References}

%@CONTROL{REVTEX41Control}
%@CONTROL{apsrev41Control,author="48",editor="1",pages="1",title="0",year="0"}

%\bibliographystyle{ieeetr}
\bibliographystyle{hieeetr}

\bibliography{all_references}

\begin{thebibliography}{10}

\bibitem{Aber06}
J.~Aberg, ``{Quantifying Superposition},'' 2006, quant-ph/0612146.

\bibitem{BaumCP14}
T.~Baumgratz, M.~Cramer, and M.~B. Plenio, ``Quantifying coherence,'' {\em
  Phys. Rev. Lett.}, vol.~113, p.~140401, 2014.

\bibitem{WintY16}
A.~Winter and D.~Yang, ``Operational resource theory of coherence,'' {\em Phys.
  Rev. Lett.}, vol.~116, p.~120404, 2016.

\bibitem{StreAP17}
A.~Streltsov, G.~Adesso, and M.~B. Plenio, ``\textit{Colloquium}: Quantum
  coherence as a resource,'' {\em Rev. Mod. Phys.}, vol.~89, p.~041003, 2017.

\bibitem{HuHPZ17}
M.-L. Hu, X.~Hu, Y.~Peng, Y.-R. Zhang, and H.~Fan, ``{Quantum coherence and
  quantum correlations},'' 2017, arXiv:1703.01852.

\bibitem{StreSDB15}
A.~Streltsov, U.~Singh, H.~S. Dhar, M.~N. Bera, and G.~Adesso, ``Measuring
  quantum coherence with entanglement,'' {\em Phys. Rev. Lett.}, vol.~115,
  p.~020403, 2015.

\bibitem{YuanZCM15}
X.~Yuan, H.~Zhou, Z.~Cao, and X.~Ma, ``Intrinsic randomness as a measure of
  quantum coherence,'' {\em Phys. Rev. A}, vol.~92, p.~022124, 2015.

\bibitem{DuBG15}
S.~Du, Z.~Bai, and Y.~Guo, ``Conditions for coherence transformations under
  incoherent operations,'' {\em Phys. Rev. A}, vol.~91, p.~052120, 2015.

\bibitem{DuBQ15}
S.~Du, Z.~Bai, and X.~Qi, ``Coherence measures and optimal conversion for
  coherent states,'' {\em Quantum Info. Comput.}, vol.~15, no.~15-16,
  pp.~1307--1316, 2015.

\bibitem{KillSP16}
N.~Killoran, F.~E.~S. Steinhoff, and M.~B. Plenio, ``Converting nonclassicality
  into entanglement,'' {\em Phys. Rev. Lett.}, vol.~116, p.~080402, 2016.

\bibitem{ZhuMCF17}
H.~Zhu, Z.~Ma, Z.~Cao, S.-M. Fei, and V.~Vedral, ``Operational one-to-one
  mapping between coherence and entanglement measures,'' {\em Phys. Rev. A},
  vol.~96, p.~032316, 2017.

\bibitem{ZhuHC17}
H.~Zhu, M.~Hayashi, and L.~Chen, ``{Axiomatic and operational connections
  between $l_1$-norm of coherence and negativity},'' 2017, arXiv:1704.02896.

\bibitem{RanaPWL16}
S.~Rana, P.~Parashar, A.~Winter, and M.~Lewenstein, ``{Logarithmic Coherence:
  Operational Interpretation of $\ell_1$-norm Coherence},'' 2017,
  arXiv:1612.09234.

\bibitem{DanaDMW17}
K.~Ben~Dana, M.~Garc\'{\i}a~D\'{\i}az, M.~Mejatty, and A.~Winter, ``Resource
  theory of coherence: Beyond states,'' {\em Phys. Rev. A}, vol.~95, p.~062327,
  2017.

\bibitem{MullDSF14}
M.~M\"uller-Lennert, F.~Dupuis, O.~Szehr, S.~Fehr, and M.~Tomamichel, ``On
  quantum {R}\'enyi entropies: A new generalization and some properties,'' {\em
  J. Math. Phys.}, vol.~54, no.~12, p.~122203, 2013.

\bibitem{WildWY14}
M.~M. Wilde, A.~Winter, and D.~Yang, ``Strong converse for the classical
  capacity of entanglement-breaking and {H}adamard channels via a sandwiched
  {R}{\'e}nyi relative entropy,'' {\em Commun. Math. Phys.}, vol.~331, no.~2,
  pp.~593--622, 2014.

\bibitem{Haya17book}
M.~Hayashi, {\em Quantum Information Theory}.
\newblock Graduate Texts in Physics, Berlin: Springer, 2017.

\bibitem{ChitG16a}
E.~Chitambar and G.~Gour, ``Comparison of incoherent operations and measures of
  coherence,'' {\em Phys. Rev. A}, vol.~94, p.~052336, 2016.

\bibitem{ShaoLLX16}
L.-H. Shao, Y.~Li, Y.~Luo, and Z.~Xi, ``Quantum coherence quantifiers based on
  the {R}\'{e}nyi $\alpha$-relative entropy,'' {\em Commun. Theor. Phys.},
  vol.~67, pp.~631--636, 2017.

\bibitem{NapoBCP16}
C.~Napoli, T.~R. Bromley, M.~Cianciaruso, M.~Piani, N.~Johnston, and G.~Adesso,
  ``Robustness of coherence: An operational and observable measure of quantum
  coherence,'' {\em Phys. Rev. Lett.}, vol.~116, p.~150502, 2016.

\bibitem{PianCBN16}
M.~Piani, M.~Cianciaruso, T.~R. Bromley, C.~Napoli, N.~Johnston, and G.~Adesso,
  ``Robustness of asymmetry and coherence of quantum states,'' {\em Phys. Rev.
  A}, vol.~93, p.~042107, 2016.

\bibitem{HoroHHH09}
R.~Horodecki, P.~Horodecki, M.~Horodecki, and K.~Horodecki, ``Quantum
  entanglement,'' {\em Rev. Mod. Phys.}, vol.~81, p.~865, 2009.

\bibitem{HayaKMM03}
M.~Hayashi, M.~Koashi, K.~Matsumoto, F.~Morikoshi, and A.~Winter, ``Error
  exponents for entanglement concentration,'' {\em J. Phys. A: Math. Gen.},
  vol.~36, no.~2, p.~527, 2003.

\bibitem{Haya06}
M.~Hayashi, ``General formulas for fixed-length quantum entanglement
  concentration,'' {\em IEEE Trans. Inf. Theory}, vol.~52, no.~5,
  pp.~1904--1921, 2006.

\bibitem{HayaZC17}
M.~Hayashi and H.~Zhu, ``{Secure uniform random number extraction via
  incoherent strategies},'' 2017, arXiv:1706.04009.

\bibitem{MosoO15}
M.~Mosonyi and T.~Ogawa, ``Quantum hypothesis testing and the operational
  interpretation of the quantum {R}{\'e}nyi relative entropies,'' {\em Commun.
  Math. Phys.}, vol.~334, no.~3, pp.~1617--1648, 2015.

\bibitem{AudeD15}
K.~M.~R. Audenaert and N.~Datta, ``$\alpha-z$-{R}\'enyi relative entropies,''
  {\em J. Math. Phys.}, vol.~56, no.~2, p.~022202, 2015.

\bibitem{Datt09}
N.~Datta, ``Min- and max-relative entropies and a new entanglement monotone,''
  {\em IEEE Trans. Inf. Theory}, vol.~55, no.~6, pp.~2816--2826, 2009.

\bibitem{Datt09b}
N.~Datta, ``Max-relative entropy of entanglement, alias log robustness,'' {\em
  I. J. Quantum Inf.}, vol.~07, no.~02, pp.~475--491, 2009.

\bibitem{DupuKFR13}
F.~Dupuis, L.~Kr\"amer, P.~Faist, J.~M. Renes, and R.~Renner, {\em Generalized
  entropies, \textrm{in} XVIIth International Congress on Mathematical
  Physics}, pp.~134--153.
\newblock World Scientific, 2013.

\bibitem{LiebT76}
E.~Lieb and W.~Thirring, {\em Studies in Mathematical Physics}.
\newblock Princeton: Princeton University Press, 1976.

\bibitem{Arak90}
H.~Araki, ``On an inequality of {L}ieb and {T}hirring,'' {\em Lett. Math.
  Phys.}, vol.~19, no.~2, pp.~167--170, 1990.

\bibitem{Hiai94}
F.~Hiai, ``Equality cases in matrix norm inequalities of {G}olden-{T}hompson
  type,'' {\em Linear and Multilinear Algebra}, vol.~36, no.~4, pp.~239--249,
  1994.

\bibitem{Petz86}
D.~Petz, ``Quasi-entropies for finite quantum systems,'' {\em Rep. Math.
  Phys.}, vol.~23, no.~1, pp.~57--65, 1986.

\bibitem{Beig13}
S.~Beigi, ``Sandwiched {R}\'enyi divergence satisfies data processing
  inequality,'' {\em J. Math. Phys.}, vol.~54, no.~12, p.~122202, 2013.

\bibitem{FranL13}
R.~L. Frank and E.~H. Lieb, ``Monotonicity of a relative {R}\'enyi entropy,''
  {\em J. Math. Phys.}, vol.~54, no.~12, p.~122201, 2013.

\bibitem{Haya17book2}
M.~Hayashi, {\em A Group Theoretic Approach to Quantum Information}.
\newblock Berlin: Springer, 2017.

\bibitem{TomaBH14}
M.~Tomamichel, M.~Berta, and M.~Hayashi, ``Relating different quantum
  generalizations of the conditional {R}\'enyi entropy,'' {\em J. Math. Phys.},
  vol.~55, no.~8, p.~082206, 2014.

\bibitem{KoniRS09}
R.~Konig, R.~Renner, and C.~Schaffner, ``The operational meaning of min- and
  max-entropy,'' {\em IEEE Trans. Inf. Theory}, vol.~55, no.~9, pp.~4337--4347,
  2009.

\bibitem{LediRD17}
F.~Leditzky, C.~Rouz{\'e}, and N.~Datta, ``Data processing for the sandwiched
  {R}{\'e}nyi divergence: a condition for equality,'' {\em Lett. Math. Phys.},
  vol.~107, no.~1, pp.~61--80, 2017.

\bibitem{ArakL70}
H.~Araki and E.~H. Lieb, ``Entropy inequalities,'' {\em Commun. Math. Phys.},
  vol.~18, no.~2, pp.~160--170, 1970.

\bibitem{MisrBPS15}
A.~Misra, A.~Biswas, A.~K. Pati, A.~Sen(De), and U.~Sen, ``Quantum correlation
  with sandwiched relative entropies: Advantageous as order parameter in
  quantum phase transitions,'' {\em Phys. Rev. E}, vol.~91, p.~052125, 2015.

\bibitem{VedrPRK97}
V.~Vedral, M.~B. Plenio, M.~A. Rippin, and P.~L. Knight, ``Quantifying
  entanglement,'' {\em Phys. Rev. Lett.}, vol.~78, no.~12, pp.~2275--2279,
  1997.

\bibitem{VedrP98}
V.~Vedral and M.~B. Plenio, ``Entanglement measures and purification
  procedures,'' {\em Phys. Rev. A}, vol.~57, no.~3, pp.~1619--1633, 1998.

\bibitem{ZhuCH10}
H.~Zhu, L.~Chen, and M.~Hayashi, ``Additivity and non-additivity of
  multipartite entanglement measures,'' {\em New J. Phys.}, vol.~12, no.~8,
  p.~083002, 2010.

\bibitem{VidaT99}
G.~Vidal and R.~Tarrach, ``Robustness of entanglement,'' {\em Phys. Rev. A},
  vol.~59, no.~1, pp.~141--155, 1999.

\bibitem{HarrN03}
A.~W. Harrow and M.~A. Nielsen, ``Robustness of quantum gates in the presence
  of noise,'' {\em Phys. Rev. A}, vol.~68, no.~1, p.~012308, 2003.

\bibitem{Stei03}
M.~Steiner, ``Generalized robustness of entanglement,'' {\em Phys. Rev. A},
  vol.~67, p.~054305, 2003.

\bibitem{Bran05}
F.~G. S.~L. Brand\~ao, ``Quantifying entanglement with witness operators,''
  {\em Phys. Rev. A}, vol.~72, no.~2, p.~022310, 2005.

\bibitem{HayaMMO06}
M.~Hayashi, D.~Markham, M.~Murao, M.~Owari, and S.~Virmani, ``Bounds on
  multipartite entangled orthogonal state discrimination using local operations
  and classical communication,'' {\em Phys. Rev. Lett.}, vol.~96, no.~4,
  p.~040501, 2006.

\bibitem{WeiG03}
T.-C. Wei and P.~M. Goldbart, ``Geometric measure of entanglement and
  applications to bipartite and multipartite quantum states,'' {\em Phys. Rev.
  A}, vol.~68, no.~4, p.~042307, 2003.

\bibitem{StreKB10}
A.~Streltsov, H.~Kampermann, and D.~Bru\ss{}, ``Linking a distance measure of
  entanglement to its convex roof,'' {\em New J. Phys.}, vol.~12, p.~123004,
  2010.

\bibitem{VidaW02}
G.~Vidal and R.~F. Werner, ``Computable measure of entanglement,'' {\em Phys.
  Rev. A}, vol.~65, p.~032314, 2002.

\bibitem{PlenVP00}
M.~B. Plenio, S.~Virmani, and P.~Papadopoulos, ``Operator monotones, the
  reduction criterion and the relative entropy,'' {\em J. Phys. A: Math. Gen.},
  vol.~33, no.~22, p.~L193, 2000.

\bibitem{HoroH99}
M.~Horodecki and P.~Horodecki, ``Reduction criterion of separability and limits
  for a class of distillation protocols,'' {\em Phys. Rev. A}, vol.~59, no.~6,
  pp.~4206--4216, 1999.

\bibitem{HayaC11}
M.~Hayashi and L.~Chen, ``Weaker entanglement between two parties guarantees
  stronger entanglement with a third party,'' {\em Phys. Rev. A}, vol.~84,
  p.~012325, 2011.

\bibitem{Rast16}
A.~E. Rastegin, ``Quantum-coherence quantifiers based on the {T}sallis relative
  $\alpha$ entropies,'' {\em Phys. Rev. A}, vol.~93, p.~032136, 2016.

\bibitem{Rain99}
E.~M. Rains, ``Bound on distillable entanglement,'' {\em Phys. Rev. A},
  vol.~60, no.~1, pp.~179--184, 1999.

\bibitem{HayaT16}
M.~Hayashi and M.~Tomamichel, ``Correlation detection and an operational
  interpretation of the {R}\'enyi mutual information,'' {\em J. Math. Phys.},
  vol.~57, no.~10, p.~102201, 2016.

\bibitem{WataH17}
S.~Watanabe and M.~Hayashi, ``Finite-length analysis on tail probability for
  {M}arkov chain and application to simple hypothesis testing,'' {\em Ann.
  Appl. Probab.}, vol.~27, no.~2, pp.~811--845, 2017.

\bibitem{Li14}
K.~Li, ``Second-order asymptotics for quantum hypothesis testing,'' {\em Ann.
  Statist.}, vol.~42, no.~1, pp.~171--189, 2014.

\bibitem{TomaH13}
M.~Tomamichel and M.~Hayashi, ``A hierarchy of information quantities for
  finite block length analysis of quantum tasks,'' {\em IEEE Trans. Inf.
  Theory}, vol.~59, no.~11, pp.~7693--7710, 2013.

\bibitem{ChenH17}
H.-C. Cheng and M.-H. Hsieh, ``Moderate deviation analysis for
  classical-quantum channels and quantum hypothesis testing,'' 2017,
  arXiv:1701.03195.

\end{thebibliography}

\end{document}